\documentclass[a4paper,UKenglish,cleveref,autoref,thm-restate,numberwithinsect]{lipics-v2021}

\pdfoutput=1 
\hideLIPIcs  

\usepackage[noadjust]{cite}
\title{Complete $\omega$-Regular Supermartingale Certificates}

\author{Alessandro Abate}{University of Oxford, UK}{alessandro.abate@cs.ox.ac.uk}{https://orcid.org/0000-0002-5627-9093}{}
\author{Mirco Giacobbe}{University of Birmingham, UK}{m.giacobbe@bham.ac.uk}{https://orcid.org/0000-0001-8180-0904}{}
\author{Sergey Ichtchenko}{University of Oxford, UK}{sergey.ichtchenko@cs.ox.ac.uk}{https://orcid.org/0009-0006-2938-3158}{}
\author{Diptarko Roy}{University of Birmingham, UK}{d.s.roy@bham.ac.uk}{https://orcid.org/0009-0003-4306-2076}{}

\authorrunning{A. Abate, M. Giacobbe, S. Ichtchenko, D. Roy} 

\Copyright{Alessandro Abate, Mirco Giacobbe, Sergey Ichtchenko, Diptarko Roy} 



\category{} 

\relatedversion{} 

\acknowledgements{This work was supported in part by the Advanced Research and Invention Agency under the Safeguarded AI programme.}

\nolinenumbers 

\newcommand{\ffilti}[1]{\mathcal{F}_i}
\def\Expect{ {\mathsf E}}
\def\Prob{ \mathsf{P} }
\def\chain{\mathbf{\Phi}}
\newcommand{\mc}[1]{\Phi_{#1}}
\newcommand{\ind}[1]{\mathbf{1}{#1}}

\def\Nat{\mathbb{N}}
\def\Real{\mathbb{R}}
\renewcommand{\d}[1]{\ensuremath{\operatorname{d}\!{#1}}}

\makeatletter
\newsavebox{\@brx}
\newcommand{\llangle}[1][]{\savebox{\@brx}{\(\m@th{#1\langle}\)}%
  \mathopen{\copy\@brx\mkern2mu\kern-0.9\wd\@brx\usebox{\@brx}}}
\newcommand{\rrangle}[1][]{\savebox{\@brx}{\(\m@th{#1\rangle}\)}%
  \mathclose{\copy\@brx\mkern2mu\kern-0.9\wd\@brx\usebox{\@brx}}}
\makeatother

\DeclareMathOperator{\Fin}{ \mathsf{Fin} }
\DeclareMathOperator{\Inf}{ \mathsf{Inf} }

\usepackage{pifont}

\usepackage{multicol}

\usepackage{amsmath,amssymb}
\usepackage{amsfonts}

\usepackage{tikz}
\usetikzlibrary{automata, positioning, fit, cd}
\usetikzlibrary[positioning,calc,arrows.meta,automata]
\usetikzlibrary{shapes, backgrounds}
\usepackage{pgfplots}
\pgfplotsset{compat=1.18}
\usepgfplotslibrary{fillbetween}

\usepackage[most]{tcolorbox}
\tcbuselibrary{breakable,skins}
\newcommand{\tightequations}{%
    \setlength{\abovedisplayskip}{3pt}%
    \setlength{\belowdisplayskip}{0pt}%
}

\makeatletter
\newcommand{\tagpad}{\def\maketag@@@##1{\hbox{\m@th\normalfont##1\hspace*{6pt}}}}
\makeatother

\tcbset{
    proofrulebox/.style={
        enhanced,
        breakable,
        arc=3pt,
        outer arc=3pt,
        boxrule=0.6pt,
        colframe=black!75,
        colback=white,
        coltitle=black,
        fonttitle=\bfseries\small,
        attach boxed title to top left={xshift=8pt,yshift=-2mm},
        boxed title style={
            enhanced,
            arc=3pt,
            outer arc=3pt,
            boxrule=0.6pt,
            colframe=black!75,
            colback=white,
            left=2pt,
            right=2pt,
            top=0pt,
            bottom=0pt,
        },
        left=8pt,
        right=8pt,
        top=8pt,
        bottom=4pt,
        before skip=8pt,
        after skip=8pt,
    },
    middleproofrulebox/.style={
        proofrulebox,
        colback=blue!5,
        boxed title style={
            enhanced,
            arc=3pt,
            outer arc=3pt,
            boxrule=0.6pt,
            left=2pt,
            right=2pt,
            top=0pt,
            bottom=0pt,
        },
    }
}

\EventEditors{Claudia Faggian and Joost-Pieter Katoen}
\EventNoEds{2}
\EventLongTitle{41st Annual Symposium on Logic in Computer Science (LICS 2026)}
\EventShortTitle{LICS 2026}
\EventAcronym{LICS}
\EventYear{2026}
\EventDate{July 20--23, 2026}
\EventLocation{Lisbon, Portugal}
\EventLogo{}
\SeriesVolume{380}
\ArticleNo{42}

\begin{document}

\maketitle

\begin{abstract}
We introduce a general methodology for the construction of sound and complete proof rules for the almost-sure and quantitative acceptance of reactivity properties on time-homogeneous Markov chains with general state spaces.
Reactivity captures the $\omega$-regular properties and subsumes linear temporal logic.
Our core technical result establishes that every reactivity property admits decomposition into multiple obligations of almost-sure termination into absorbing regions, and that appropriate absorbing regions always exist on general state spaces. This enables the extension of every complete proof rule for almost-sure termination into a proof rule for reactivity that is complete in the almost-sure case, and complete up to an arbitrarily small $\varepsilon$-approximation in the quantitative case.
We apply our new methodology to recent results on sound and complete supermartingale certificates for almost-sure termination in the special case of countably infinite state spaces, alongside standard results on quantitative safety. As a result, we obtain the first sound and complete supermartingale certificates for almost-sure $\omega$-regular properties and the first sound and $\varepsilon$-complete supermartingale certificates for quantitative $\omega$-regular properties on time-homogeneous Markov chains with countably infinite state spaces.
\end{abstract}

\section{Introduction}
Supermartingale certificates are essential components of widely used proof rules for the formal verification of probabilistic systems whose semantics correspond to time-homogeneous Markov chains over general state spaces. Examples include probabilistic programs used to model randomised algorithms and inference protocols, which typically operate over countably infinite state spaces, as well as stochastic difference equations used to model discrete-time dynamical systems, which usually evolve over continuous state spaces. In many real-world domains, the finite-state assumptions underlying traditional probabilistic model-checking algorithms are inappropriate, making it crucial to model and reason about probabilistic systems over general measurable state spaces.

Deductive approaches based on martingale theory for reasoning about probabilistic systems bring principles from stochastic analysis into program analysis and control theory. Supermartingales are real-valued stochastic processes whose values satisfy specific convergence properties; supermartingale certificates map a system to a supermartingale and are essentially the probabilistic analogues of ranking functions and inductive invariants in Hoare logic, as well as Lyapunov functions in stability analysis~\cite{kushner1967stochastic,DBLP:conf/cav/ChakarovS13}.
In this paper, we bring together stochastic analysis and formal language theory, providing the technical elements for the construction of sound and complete supermartingale certificates for reactivity properties.

Reactivity is the temporal property stating that either an event 
$A$ occurs only finitely many times or an event 
$B$ occurs infinitely often or, more generally, it is a finite conjunction of such properties.
It is a standard result that reactivity characterises the $\omega$-regular languages, 
thereby subsuming linear temporal logic and generalising properties such as recurrence, persistence, safety (or invariance), and termination (or reachability)~\cite{DBLP:conf/podc/MannaP89, Safra88}. Proof rules based on supermartingale certificates for the almost-sure and quantitative acceptance of these classes of properties---including reactivity---have been proposed previously~\cite{McIverMKK18,DBLP:journals/toplas/TakisakaOUH21,DBLP:conf/cav/AbateGR24}; however, for reactivity, these proof rules provide only sufficient conditions, meaning they may fail to certify systems that nonetheless satisfy the property. 

The completeness of a proof rule involving supermartingale certificates concerns whether {\em every} stochastic process that satisfies the corresponding temporal property admits a certificate---typically expressed as a value function---that fulfils the associated proof rule. Complete proof rules for reactivity \cite[Definition 4.4]{DBLP:journals/corr/abs-2512-00270} have been presented recently, although under the assumption that value functions may express arbitrary distributions. While this enables reasoning about the distribution of the number of visits to states, it also means that it is hard to automate the construction of such proof certificates~\cite[p.\ 11]{DBLP:journals/corr/abs-2512-00270}.  
In fact, the existing certificate construction methods require that proof certificates are scalar-valued functions.
In this work, we focus on the setting of certificates that are (i) complete and (ii) scalar-valued. Prior to the present work, proof rules satisfying both criteria were known only for almost-sure termination~\cite{McIverMKK18,MajumdarS25}; completeness for reactivity and therefore $\omega$-regular properties remained an open problem.

We introduce a general methodology for characterising sound and complete proof rules for reactivity properties. Our core technical result shows that every reactivity property admits a decomposition into (i) multiple obligations of qualitative safety for each {\em absorbing region}, (ii) multiple obligations of almost-sure termination into the absorbing regions,  and (iii) one obligation of quantitative safety for the overall system. In other words, our methodology reduces the problem of establishing almost-sure or quantitative reactivity to a collection of almost-sure termination and qualitative or quantitative safety problems.
Crucially, we show that every time-homogeneous Markov chain on a general state space that satisfies a reactivity property admits suitable absorbing regions that may be used to approximate the probability of acceptance up to any arbitrary $\varepsilon$-approximation. Moreover, when the reactivity property holds almost surely, we show that there exist suitable absorbing regions that exactly characterise almost-sure acceptance.

Our result introduces the first sound and complete proof rules yielding supermartingale certificates for reactivity, and thus for $\omega$-regular properties. 
We present two such proof rules, obtained by applying our methodology to complete proof rules for almost-sure termination on the special case of countably infinite state spaces~\cite{McIverMKK18,MajumdarS25}, alongside standard results on quantitative safety~\cite{kushner1967stochastic,DBLP:journals/toplas/TakisakaOUH21}. As an immediate consequence of the completeness results underlying our absorbing-region decomposition, these new proof rules are complete on countably infinite state spaces up to arbitrary $\varepsilon$-ap\-prox\-i\-ma\-tion for quantitative reactivity, and fully complete when the reactivity property holds almost surely.

Our contribution is twofold. First, we show that every reactivity property admits a decomposition into multiple obligations of almost-sure termination and quantitative safety with respect to absorbing regions. Our decomposition is $\varepsilon$-complete in the general case and fully complete in the almost-sure case for any time-homogeneous Markov chain on a general state space. 
Second, we instantiate our decomposition using existing complete proof rules for almost-sure termination over countably infinite state spaces. As a result, we obtain the first supermartingale certificates that are $\varepsilon$-complete in the quantitative case and fully complete in the almost-sure case for 
$\omega$-regular properties on countably infinite state spaces.

\section{Markov Chains and Probabilistic Reactivity}

We consider stochastic processes on a general state space $(S, \Sigma)$, where $S$ denotes the set of states
and the associated $\sigma$-algebra $\Sigma$ denotes the set of measurable regions. We treat the problem of determining the acceptance of  
infinite time horizon properties measured over $(\Omega, {\mathcal F})$, 
where the set of outcomes $\Omega = S^\omega$ are the infinite trajectories of a process 
and the set of events ${\mathcal F} = \bigotimes_{i \in \omega} \Sigma_i$ (with $\Sigma_i = \Sigma$) 
are the measurable properties \cite[Chapter 3.4]{meyn_tweedie_glynn_2009}. 
A filtration $\mathfrak{F} = ( \mathcal{F}_0, \mathcal{F}_1, \ldots,\mathcal{F}_n, \ldots)$ is an increasing sequence of sub-$\sigma$-algebras of $\mathcal{F}$ where every $L \in \mathcal{F}_n$ is an event whose membership is determined by the prefix of length $n$.
A sequence of random variables $\chain = (\mc{0}, \mc{1}, \ldots, \allowbreak \mc{n}, \ldots)$ is a stochastic process adapted to $\mathfrak{F}$ precisely if $\mc{n}$ is $\mathcal{F}_{n}$-measurable for every $n \in \Nat$.

Throughout the paper, we use the notation $\ind\{ \pi \} :  S^\omega \to \{ 0, 1 \} $ for the indicator function of the proposition $\pi$, 
we use $\{ \pi \} $ to denote the subset of $S^\omega$ consisting of trajectories that satisfy proposition $\pi$,
and use $\delta_s : \Sigma \to [0, 1]$ for the Dirac measure $\delta_s(A)=\ind\{s\in A\}$ at $s \in S$. We say that two events are {\em almost-surely equivalent} under a probability measure precisely if their symmetric difference is a probability zero event.

We restrict our attention to time-homogeneous Markov chains on general state spaces. 
As is standard in stochastic analysis, we rely on the result that an initial probability measure 
and a transition probability kernel define a unique 
probability measure over the measurable space of infinite trajectories and properties.

\begin{theorem}[{\protect \cite[Theorem 3.4.1]{meyn_tweedie_glynn_2009}}]
\label{thm:meas-spec}
Let $\mu : \Sigma \to [0, 1]$ be an initial probability measure and 
$P : S \times \Sigma \to [0, 1]$ be a transition probability kernel. Then, there exists a stochastic process  $\mathbf{\Phi} = ( \Phi_0, \Phi_1, \ldots )$ adapted to $\mathfrak{F}$ 
and a probability measure ${\sf P}_\mu : \mathcal{F} \to [0,1]$ where ${\sf P}_\mu(\mathbf{\Phi} \in L)$ is the probability that $\chain$ satisfies the property $L \in {\mathcal F}$ and, for every $n \in \Nat$ and 
$A_0 \in \Sigma, \ldots, A_n \in \Sigma$, 
\begin{multline}
        {\sf P}_\mu(\Phi_0 \in A_0 \land \dots \land \Phi_n \in A_n)
    =\\\int_{s_0 \in A_0} \cdots \int_{s_{n-1} \in A_{n-1}} 
    \mu(\mathrm{d}s_0)
    P(s_0, \mathrm{d}s_1) 
    \cdots 
    P(s_{n - 1}, A_n).\label{eqn:measure-traj}
\end{multline}
\end{theorem}

\begin{definition}[Time-Homogeneous Markov Chain]
    A time-ho\-mo\-ge\-neous Markov chain defined by an initial probability measure $\mu : \Sigma \to [0, 1]$ and
    the transition probability kernel $P : S \times \Sigma \to [0, 1]$ is a stochastic process satisfying \cref{eqn:measure-traj}.
\end{definition}

We consider $\omega$-regular properties which we define in terms of reactivity, which in turn we define in terms of visiting a region $A$ finitely or infinitely many times, respectively, as follows: 
\begin{align}
    &\Fin(A) = \bigcup_{n \in \Nat} \bigcap_{m \geq n} \{ \mc{m} \notin A \},
    &&\Inf(A) = \bigcap_{n \in \Nat} \bigcup_{m \geq n} \{ \mc{m} \in A \}.
    \label{eqn:inf-definition}
\end{align}
    
Using these components, we define reactivity as an intersection of events of the form $\Fin(A_i) \cup \Inf(B_i)$  for $i = 1, \dots, k$ on a finite sequence 
of pairs $(A_1, B_1), \dots, (A_k, B_k) \in \Sigma^2$, which we refer to as a Streett acceptance condition.
\begin{definition}[Reactivity]\label{def:reactivity}
A reactivity property defined by the Streett acceptance condition $(A_1, B_1) \in \Sigma^2, \dots, (A_k, B_k) \in \Sigma^2$
is the event
\begin{equation}\label{eqn:react-property}
    \bigcap_{i = 1}^k \Fin (A_i) \cup \Inf (B_i).
\end{equation}
\end{definition}

\begin{remark}[$\omega$-Regular Properties]\label{rem:omega-regular}
Reactivity conditions are sufficient to encode $\omega$-regular properties. 
Every $\omega$-regular property $\varphi$ over a finite set of atomic propositions $\Pi$ 
may be represented as the set of accepting runs of a 
corresponding deterministic Streett automaton (DSA) \cite{Safra88}. 
A DSA $\mathcal{A}$ consists of 
a finite set of states $Q$, 
an initial state $q_0 \in Q$, 
a transition function $T \colon Q \times 2^\Pi \to Q $, 
and an acceptance condition $(F_1, G_1), \dots, (F_k,G_k)$ where $F_i, G_i \subseteq Q$ for $i = 1, \dots, k$.
An infinite input trace $(p_0, p_1, \allowbreak p_2, \dots ) \in (2^\Pi)^\omega$ is accepted if there exists 
an infinite run $(q_0, q_1, \allowbreak q_2, \dots) \in Q^\omega$ such that $q_{n+1} = T(q_n, p_n)$ for every $n \in \Nat$ and,  
for every $i = 1, \dots, k$, either $\sum_{n=0}^\infty \ind{ \{ q_n \in F_i \} } < \infty$ or $\sum_{n=0}^\infty \ind{\{q_n \in G_i \} } = \infty$.

Given a time-homogeneous Markov chain $\hat \chain$ defined by $(\hat \mu, \hat P)$ with state space $(\hat S, \hat \Sigma)$, 
we define the synchronous composition of $\hat{\chain}$ and the automaton $\mathcal{A}$ as another time-homogeneous Markov chain 
$\chain$ with state space $(\hat S \times Q, \hat \Sigma \otimes 2^Q)$, defined by the initial distribution $\mu(A) = \int_{ s \in \hat{S} } 
\ind{ \{ (s, q_0) \in A \} } ~\hat{\mu}(\mathrm{d}s)$,
and the transition probability kernel 
\begin{align}
P((s, q), A)
= \int_{(s^\prime,q^\prime) \in A}
\ind{ 
\{ T(q, \llangle s \rrangle) = q^\prime \}
}
~\hat{P}(s, \mathrm{d}s^\prime),
\end{align} 
where $\llangle \cdot \rrangle \colon \hat S \to 2^\Pi$ denotes the observation (or labelling) function of $\hat \chain$, 
yielding the set of propositions $\llangle s \rrangle$ that hold true in $s$. 
We then define the Streett acceptance condition $(A_1, B_1), \dots, (A_k, B_k)$ 
where $A_i = \hat S \times F_i$ and $B_i = \hat S \times G_i$ for $i = 1, \dots, k$. 
Then, this construction yields an equivalence   
between the probabilities of the original process $\hat \chain$ satisfying the $\omega$-regular property $\varphi$ and 
the synchronous product $\chain$ satisfying the derived reactivity property:
\begin{equation}
    \Prob_{\hat \mu} (\hat \chain \in \varphi) = \Prob_\mu\left(\chain \in \bigcap_{i = 1}^k \Fin (A_i) \cup \Inf (B_i)\right). 
\end{equation}
For this reason, henceforth we consider reactivity properties, without loss of generality for $\omega$-regular properties. \lipicsEnd
\end{remark}

\begin{figure}[t]
\centering
    \begin{tikzpicture}[minimum size=7mm, node distance=14mm]
        \node[draw, circle] (1) {1};

        \node[draw, circle, above right of=1] (2) {2};
        \node[draw,circle,right of=2] (3) {3};
        \node[draw,circle,right of=3] (4) {4};
        \node[right of=4] (5) {};

        \node[draw, circle, below right of=1] (0) {0};
        \node[draw, circle, right of=0] (m1) {$-1$};
        \node[draw, circle, right of=m1] (m2) {$-2$};
        \node[right of=m2] (m3) {};
        \draw (1) edge[->, bend left] node[above] {$\frac{1}{2}$} (2)
        (2) edge[->, bend left] node[below,xshift=2mm,yshift=3mm] {$\frac{1}{2}$} (1)
        (2) edge[->,bend left] node[above] {$\frac{1}{2}$} (3)
        (3) edge[->,bend left] node[below] {$\frac{1}{2}$} (2)
        (3) edge[->,bend left] node[above] {$\frac{1}{2}$} (4)
        (4) edge[->,bend left] node[below] {$\frac{1}{2}$} (3)
        (4) edge[->, bend left, dashed] node[above] {$\frac{1}{2}$} (5)
        (5) edge[->, bend left, dashed] node[below] {$\frac{1}{2}$} (4);
        \def\pPlus{$\frac{1+\varepsilon}{2}$}
        \def\pMinus{$\frac{1-\varepsilon}{2}$}
        \draw (1) edge[->, bend left] node[above,xshift=3mm,yshift=-3.3mm] {$\frac{1}{2}$} (0);
        \draw (0) edge[->, bend left] node[below] {$\frac{1}{2}$} (1);
        \draw (0) edge[->, bend left] node[above] {$\frac{1}{2}$} (m1);
        \draw (m1) edge[->, bend left] node[below] {\pMinus} (0);
        \draw (m1) edge[->, bend left] node[above] {\pPlus} (m2);
        \draw (m2) edge[->, bend left] node[below] {\pMinus} (m1);
        \draw
        (m2) edge[->, bend left, dashed] node[above] {\pPlus} (m3)
        (m3) edge[->, bend left, dashed] node[below] {\pMinus} (m2);
        \node[right=12mm of 5, anchor=west, yshift=-7mm] (profit)
            {$\textit{Profit} = \{ w : w > 0 \}$};
        \node[anchor=north west,yshift=2mm] (solvency) at (profit.south west)
            {$\textit{Solvency} = \{ w : w \geq 0 \}$};
        \node[anchor=north west,yshift=2mm] (debt) at (solvency.south west)
            {$\textit{Debt} = \{ w : w < 0 \}$};
    \end{tikzpicture}
    \caption{The Lending Casino}\label{fig:lending-casino}
\end{figure}

\begin{example}[The Lending Casino, {\protect \cref{fig:lending-casino}}]\label{ex:lending-casino-intro}

A casino decides on the following business strategy. Each gambler is provided \$1 of initial wealth. A fair coin is tossed, and the gambler's wealth either increases or decreases by \$1 depending upon the outcome of the toss.
If the gambler's wealth becomes negative, then they are given the option of borrowing from the casino, but while they are in debt, the coin is biased such that the probability of winning \$1 is $\frac{1-\varepsilon}{2}$ for some $\varepsilon > 0$.
Notably, there is no upper limit to either the gambler's wealth or the size of their overdraft, and there is no upper bound to the duration of the game.

The casino wishes to answer the following question regarding infinite trajectories: \textit{will a gambler eventually fall into debt and remain there forever?}
This corresponds to the requirement that the property $\Inf(\textit{Solvency})$ holds with probability zero, or equivalently $\Fin(\textit{Solvency})$ holds with probability one. This is the reactivity property defined by the Streett acceptance condition given by the single pair $(\textit{Solvency}, \emptyset)$. \lipicsEnd
\end{example}

A stopping time $\zeta : \Omega \to \Nat \cup \{ \infty\} $ is a random variable for which the event $\{ \zeta = n \}$ is an element of $\mathcal{F}_n$, for every $n \in \Nat$. A special case of a stopping time is the first hitting time of a region, which is key to defining the notion of return times in \cref{sec:decomp} and of termination in \cref{sec:complete-certificates}.
\begin{definition}[First Hitting Time]\label{def:hitting-time} The first hitting time $\tau_A$ of a region $A \in \Sigma$ by the stochastic process $\chain$ is
\begin{align}
    \tau_A = \inf \{ n \geq 0 \colon \mc{n} \in A\}.
\end{align}
\end{definition}
 
\begin{example}[Profit Leads to Debt, and Debt Doesn't Always Lead to Solvency]
We use the notion of first hitting time to describe the event that the gambler hits \textit{Debt} as $\tau_\textit{Debt} < \infty$. It is a classic result \cite[p.182]{durrett2019probability} that a symmetric random walk on the integers visits every integer with probability 1, which implies that from any initial positive wealth $w \in \textit{Profit}$ that the gambler eventually enters debt:
\begin{equation}
    \forall w \in \textit{Profit} \mathpunct.
    \Prob_{\delta_w}( \tau_\textit{Debt} < \infty )= 1,
\end{equation}
and $\forall w \in \textit{Profit} \mathpunct. \Expect_{\delta_w}[\tau_\textit{Debt}] = \infty$ \cite[p.304]{durrett2019probability}.
However, when the gambler reaches a wealth of $-\$1$ or smaller, their probability of returning to \textit{Solvency} is less than 1:
\begin{align}\label{eqn:lending-casino-hitting-probability}
    \sup_{w \in \textit{Debt}}
    \Prob_{\delta_w}( \tau_{\textit{Solvency}} < \infty  ) = 
    \sup_{w \in \textit{Debt}} 
    \left(\frac{1-\varepsilon}{1+\varepsilon}\right)^{|w|}
    = 
    \left(\frac{1-\varepsilon}{1+\varepsilon}\right) < 1.\qquad\qquad\qquad~~~~\lipicsEnd 
\end{align}
\end{example}

In order to reason about the interaction of stopping times with the Markov property, 
we recall the (time) shift operator $\theta : \Omega \to \Omega$ \cite[Definition 3.1.8, p.58]{douc2018markov} defined by
\begin{equation}
    \theta(s_0, s_1, \ldots, s_n, \ldots)
    =
    ( s_1, s_2, \ldots, s_{n+1}, \ldots )
\end{equation}
that deletes the first element of a trajectory. Given $n \in \Nat$, we can apply the shift operator $n$ times to delete a prefix of length $n$, denoted as $\theta^n$ in the usual way. 
The shift operator gives rise to two fundamental concepts in stochastic analysis: shift invariance and the Markov property. 
Our technical results on completeness for quantitative reactivity relies on the fact that reactivity is shift-invariant, i.e., invariant to addition and removal of arbitrary prefixes.

\begin{theorem}[Reactivity is Shift-Invariant]\label{thm:reactivity-shift-invariant}
Let $(A_1, B_1) \in \Sigma^2, \dots, (A_k, B_k) \in \Sigma^2$ be a Streett acceptance condition, and let $L$ be 
the corresponding reactivity property defined in \cref{eqn:react-property}. Then,
\begin{equation}\label{eqn:shift-inv-eqn}
    \theta^{-1} L = L. 
\end{equation}
\end{theorem}
\begin{proof}
The event of visiting a region $A$ finitely many times $\Fin(A)$, as well as the event of visiting a 
region $B$ infinitely many times $\Inf(B)$ both satisfy shift invariance:
\begin{align}
    \notag\theta^{-1}\Fin(A) &= \{w \in \Omega \colon \exists m\geq 1 \mathpunct. \forall n\geq m \mathpunct. \mc{n}(w)\notin A\} \\&= \{w \in \Omega \colon \exists m\geq 0 \mathpunct. \forall n\geq m \mathpunct. \mc{n}(w)\notin A\} = \Fin(A),\\
    \notag\theta^{-1}\Inf(B) &= \{w \in \Omega \colon \forall m\geq 1 \mathpunct. \exists n\geq m \mathpunct. \mc{n}(w)\in B\}\\ &=  \{w \in \Omega \colon \forall m\geq 0 \mathpunct. \exists n\geq m \mathpunct. \mc{n}(w)\in B\} = \Inf(B).
\end{align}
The class of events $L$ for which \cref{eqn:shift-inv-eqn} holds is closed under countable unions and intersections (\cite[Proposition 5.1.5(i), p.99]{douc2018markov}).
Since the reactivity property \cref{eqn:react-property} is a finite Boolean combination of $\Fin$- and $\Inf$-events, it satisfies \cref{eqn:shift-inv-eqn} as well.
\end{proof}

Our technical results on completeness for Streett acceptance (\cref{thm:eps-completeness,thm:completeness}) build upon the result that 
every time-homogeneous Markov chain satisfies the Markov property.

\begin{theorem}[Markov Property {\protect \cite[p.~66]{meyn_tweedie_glynn_2009}}]\label{thm:strong-markov-property}
Let $\chain$ be a time-homogeneous Markov chain. Then, for 
any bounded random variable $H$ on $(\Omega, \mathcal{F}, \Prob_\mu)$, and any time $n \in \Nat$: 
\begin{align}\label{eqn:strong-markov-property}
    &\Expect_\mu
    \left[ 
    H \circ \theta^n 
    \mid \mathcal{F}_n 
    \right]
    = \Expect_{\delta_{\mc{n}}}[H]
    &&\Prob_\mu\text{-a.s.}
\end{align}
\end{theorem}


\section{Absorbing-Region Decomposition for Almost-Sure and Quantitative Reactivity}\label{sec:decomp}

We show that every reactivity property admits a decomposition into multiple independent obligations, one for each Streett pair, where each obligation consists of two interdependent components: quantitative safety (i.e.,\ safety with positive probability) and almost-sure termination into appropriate regions of the state space, termed absorbing regions. 

We first recall the notion of first return time, which is the first hitting time that occurs strictly after the initial step of the process.
\begin{definition}[First Return Time]\label{def:successive-return-times}
The first return time $\sigma_A$ of a region $A \in \Sigma$ by the stochastic process $\chain$ is
\begin{align}
    \sigma_A &= 1 + \tau_A \circ \theta, 
    \label{eqn:first-return-time}
\end{align}
where $\tau_A$ is the first hitting time of the region $A$ (\cref{def:hitting-time}).
\end{definition}

We observe that, as time tends to infinity, the probability of returning to a region 
$A$ converges almost surely to the probability of visiting 
$A$ infinitely often, and that this limit necessarily takes a value of either 0 or 1 for every initial distribution.
\begin{lemma}[{\protect \cite[Equation 4.2.5]{douc2018markov}}]\label{lem:technical}
Let $A \in \Sigma$ be a region. Then
\begin{align}
    \lim_{n \to \infty}
    \Prob_{\delta_{\mc{n}}}
    (~\sigma_A < \infty~) = \ind{ \{\chain \in \Inf(A) \}},
\label{eqn:inf-lim}
\end{align}
holds $\Prob_\xi$-a.s, for any initial probability distribution $\xi$ on $(S, \Sigma)$.
\end{lemma}

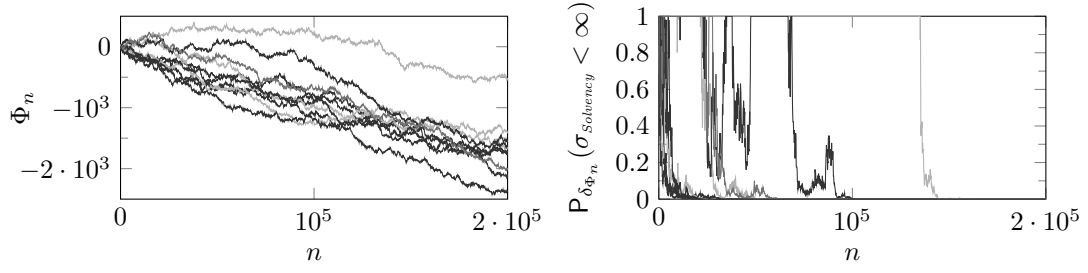
\begin{figure}
    \centering
    \begin{tikzpicture}
        \begin{axis}[
            xlabel={$n$},
            ylabel={$\Phi_n$},
            name=stateplot,
            ymin=-2500, ymax=500,
            xmin=0, xmax=200001,
            xtick distance=1000,
            ytick distance=1000,
            minor y tick num=1,
            scaled x ticks=false,
            xtick={0,100000,200000},
            height=4cm, 
            width=.5*\columnwidth-3mm,
            ytick={0,-1000, -2000},
            yticklabels={0, \(-10^3\), \(-2 \cdot 10^3\)},
            xticklabels={0, \(10^5\), \(2 \cdot 10^5\)},
            ylabel style={yshift=-12pt},
            ]
            \def\colorone{black!80}
            \def\colorzero{black!30}
            \def\colorbetween{black!55}
            \def\filename{figs/casino_state.txt}

            \addplot [\colorone] table[x={n}, y={s1}] {\filename};
            \addplot [\colorone] table[x={n}, y={s2}] {\filename};
            \addplot [\colorone] table[x={n}, y={s3}] {\filename};
            \addplot [\colorzero] table[x={n}, y={s4}] {\filename};
            \addplot [\colorzero] table[x={n}, y={s5}] {\filename};
            \addplot [\colorone] table[x={n}, y={s6}] {\filename};
            \addplot [\colorzero] table[x={n}, y={s7}] {\filename};
            \addplot [\colorone] table[x={n}, y={s8}] {\filename};
            \addplot [\colorbetween] table[x={n}, y={s9}] {\filename};
            \addplot [\colorone] table[x={n}, y={s10}] {\filename};
        \end{axis}
        
        \begin{axis}[xlabel={$n$}, ylabel={$\Prob_{\delta_{\Phi_n}}(\sigma_{\textit{\tiny Solvency}} < \infty)$},
            ymin=0, ymax=1,
            xmin=0, xmax=200001,
            xtick distance=10000,
            ytick distance=0.2,
            minor y tick num=1,
            scaled x ticks=false,
            xtick={0,100000,200000},
            height=4cm, 
            width=.5*\columnwidth-3mm,
            xticklabels={0, \(10^5\), \(2 \cdot 10^5\)},
            at={(stateplot.north east)},
            anchor=north west,
            xshift=20mm
            ]
            \def\colorone{black!80}
            \def\colorzero{black!30}
            \def\colorbetween{black!55}
            \def\filename{figs/casino_prob.txt}

            \addplot [\colorone] table[x={n}, y={s1}] {\filename};
            \addplot [\colorone] table[x={n}, y={s2}] {\filename};
            \addplot [\colorone] table[x={n}, y={s3}] {\filename};
            \addplot [\colorzero] table[x={n}, y={s4}] {\filename};
            \addplot [\colorzero] table[x={n}, y={s5}] {\filename};
            \addplot [\colorone] table[x={n}, y={s6}] {\filename};
            \addplot [\colorzero] table[x={n}, y={s7}] {\filename};
            \addplot [\colorone] table[x={n}, y={s8}] {\filename};
            \addplot [\colorbetween] table[x={n}, y={s9}] {\filename};
            \addplot [\colorone] table[x={n}, y={s10}] {\filename};
        \end{axis}
    \end{tikzpicture}
    \caption{Simulation of Return Probabilities}
    \label{fig:return_prob}
\end{figure}
\begin{example}[Simulating \cref{lem:technical}]
\Cref{fig:return_prob} empirically illustrates this standard result using 10 stochastic simulations of the Lending Casino from \cref{fig:lending-casino}. The left hand plot shows that, over long runs, the process tends to remain in {\em Debt} states (negative values) and stops visiting {\em Solvency} states (non-negative values). The right hand plot, based on the same simulations, shows that the probability of returning to \textit{Solvency} disappears from the plot, because it tends to 0. This supports the following almost-sure property of the Lending Casino:
\begin{equation}
    \ind{ \{\chain \in \Inf(\textit{Solvency}) \}} = 0 =
    \lim_{n \to \infty}
    \Prob_{\delta_{\mc{n}}}
    (~\sigma_{\textit{Solvency}} < \infty~).\qquad\qquad\lipicsEnd
\end{equation}
\end{example}

\Cref{lem:technical} leads to the following standard result---Orey's Theorem---which provides a sufficient condition for the reactivity property defined by a single Streett pair. 
\begin{theorem}[Orey's Theorem {\protect {\cite[Theorem 4.2.6]{douc2018markov}}}]\label{thm:orey}
Let $A, B \in \Sigma$ be regions such that 
\begin{align}
    \inf_{s \in A} \Prob_{\delta_s} \left(~ 
    \sigma_B < \infty
    ~\right) > 0.\label{eq:orey-condition}
\end{align}
Then $\Prob_{\mu}\left(
\Fin(A) \cup \Inf(B)
\right) = 1$.
\end{theorem}

Intuitively, the theorem states that if every state in 
$A$ leads to a trajectory that eventually returns to 
$B$ with positive probability, then any process that visits 
$A$ infinitely often must almost surely also visit 
$B$ infinitely often. In other words, it provides a sufficient condition for the event 
$\Fin(A) \cup \Inf(B)$, that is, the acceptance of a single Streett pair $(A,B)$.
However, Orey’s theorem is incomplete, because a Markov chain can almost-surely satisfy a Streett pair, yet fail to satisfy the condition \eqref{eq:orey-condition}. We will now present two counterexamples to its completeness and introduce strategies to address them. Our three main results in \cref{thm:absorbing-regions-decomposition,thm:eps-completeness,thm:completeness} then generalise these two strategies, leading to our framework for the construction of sound and complete proof rules for reactivity properties on general state spaces.

\begin{example}[Orey's Theorem is Incomplete for Almost-Sure Reactivity]\label{ex:orey-incomplete}
We show that it is possible for a Markov chain to satisfy a Streett pair almost-surely without satisfying the premise of Orey's Theorem.
Observe the finite-state system given in \cref{fig:orey_no}, and consider the Streett pair $(A,B)$ 
    where $A = \{ s_1, s_3\}$ are the states labelled by $a$, and $B = \{ s_2\}$ are the states labelled by $b$. 
    It is easy to see that the system satisfies the Streett condition $(A,B)$ almost surely, as we either alternate between $s_1$ and $s_2$, visiting $B$ infinitely often, or visit $s_3$ once and then stay in $s_4$, visiting $A$ at most finitely many times. However, it is also easy to see that the premise of Orey's Theorem (\cref{eq:orey-condition}) does not hold: consider the problematic state $s_3$. For this state, $s_3\in A$, but it is impossible to reach a state in $B$, 
    making the minimum probability to reach $B$ from $A$ nil and violating \cref{eq:orey-condition}.
    We address the incompleteness of this example by extending the set of states that the system 
    must return to infinitely often. In other words, we extend $B$ with the additional 
    set $J = \{ s_4 \}$ which, in this example, consists of all states outside of $A$ from which the probability of returning to $A$ is no greater than 1/2.
    We then observe that while \cref{eq:orey-condition} does not hold for the original pair $(A,B)$, it does hold on the extended pair $(A,B\cup J)$. Furthermore, we can intuitively see that it is sound to use this new extended Streett pair, because in this example, it is impossible to alternate 
    between $A$ and $J\setminus B$ infinitely many times. Therefore, visiting $J\setminus B$ infinitely often implies visiting $A$ finitely many times.\lipicsEnd
\end{example}
\begin{figure}[t]
    \centering
    \begin{minipage}{0.48\textwidth}
        \centering
        \begin{tikzpicture}[minimum size=7mm, node distance=14mm]
            \node[draw, circle, initial, initial text=$\mu$] (s0) {$s_0$};
            \node[draw, circle, right of=s0,
                  label={[label distance=-5pt]above:{$\{a\}$}}] (s1) {$s_1$};
            \node[draw, circle, right of=s1,
                  label={[label distance=-5pt]above:{$\{b\}$}}] (s2) {$s_2$};
            
            \node[circle, minimum size=7mm, above of=s1,
                  label={[label distance=-5pt]above:{\phantom{$\{a\}$}}}] (spacer) {};
            \node[draw, red, circle, below of=s1,
                  label={[red,label distance=-5pt]above:{$\{a\}$}}] (s3) {$s_3$};
            \node[draw, circle, right of=s3] (s4) {$s_4$};
            \draw (s0) edge[->] node[below, yshift=1mm] {$\frac{1}{2}$} (s1)
                  (s1) edge[->, bend left] node[above, yshift=-1mm] {} (s2)
                  (s2) edge[->, bend left] node[below, yshift=1mm] {} (s1)
                  (s0) edge[->] node[below left, xshift=2mm, yshift=2mm] {$\frac{1}{2}$} (s3)
                  (s3) edge[->] node[above, xshift=-1mm, yshift=-1mm] {} (s4)
                  (s4) edge[->, loop right] node[right] {} (s4);
            \node (J) [draw=orange, fit=(s4), inner sep=0.2cm, dashed, thick,
                       fill=orange!20, fill opacity=0.2] {};
            \node [xshift=2.0ex, yshift=-4.0ex, orange] at (J.east) {$J$};
        \end{tikzpicture}
        \caption{Almost-Sure Reactivity}
        \label{fig:orey_no}
    \end{minipage}\hfill
    \begin{minipage}{0.48\textwidth}
        \centering
        \begin{tikzpicture}[minimum size=7mm, node distance=14mm]
            \node[draw, circle, initial, initial text=$\mu$] (1) {$s_0$};
            \node[draw, circle, right of=1,
                  label={[label distance=-5pt]above:{$\{a\}$}}] (2) {$s_1$};
            \node[draw, circle, right of=2,
                  label={[label distance=-5pt]above:{$\{b\}$}}] (3) {$s_2$};
            \node[draw, circle, below of=2,
                  label={[label distance=-5pt]above:{$\{a\}$}}] (0) {$s_3$};
            \node[draw, circle, right of=0] (m1) {$s_4$};
            \node[draw, red, circle, above of=2,
                  label={[label distance=-5pt,red]above:{$\{a\}$}}] (5) {$s_5$};
            \draw (1) edge[->] node[below, yshift=1mm] {$\frac{1}{3}$} (2)
                  (2) edge[->, bend left] node[above] {} (3)
                  (3) edge[->, bend left] node[below] {} (2)
                  (1) edge[->] node[below left, xshift=2mm, yshift=2mm] {$\frac{1}{3}$} (0)
                  (0) edge[->] node[above] {} (m1)
                  (1) edge[->] node[above left, xshift=2mm, yshift=0mm] {$\frac{1}{3}$} (5)
                  (5) edge[->, loop right] (5)
                  (m1) edge[->, loop right] (m1);
            \node (J) [draw=orange, fit=(m1), inner sep=0.2cm, dashed, thick,
                       fill=orange!20, fill opacity=0.2] {};
            \node [xshift=3.1ex, yshift=-4.0ex, orange] at (J.east) {$J$};
            \node (I) [draw=yellow!10!green, fit=(m1) (0) (1) (2) (3),
                       inner sep=0.4cm, dashed, thick, fill=green!20, fill opacity=0.2] {};
            \node [xshift=2.0ex, yellow!10!green] at (I.east) {$I$};
        \end{tikzpicture}
        \caption{Quantitative Reactivity}
        \label{fig:orey_no_quant}
    \end{minipage}
\end{figure}

\begin{example}[Orey's Theorem is Incomplete for Quantitative Reactivity]\label{ex:orey-no-quant}
    Consider the Markov chain given in \cref{fig:orey_no_quant}, which extends the previous process with the additional problematic state $s_5$. 
    This state fails to satisfy the reactivity property for the Streett condition $(A,B)$, where $A =\{ s_1, s_3, s_5 \}$ and $B = \{ s_2 \}$, but such trajectories account only for 1/3 of the total probability. The remainder of the system satisfies the property and these trajectories constitute the remaining 2/3 of the probability mass. However, the strategy of \cref{ex:orey-incomplete} alone fails to capture this quantitative distinction. In fact, the system does not satisfy the updated condition 
    $\inf_{s \in A} \Prob_{\delta_s}(~\sigma_{B \cup J} < \infty~) > 0$, because $s_5 \in A$ has no path that can reach $B \cup J$.
    
    To resolve this, we identify an additional invariant region $I$, such that any trajectory that remains entirely within $I$ satisfies the Streett acceptance condition $(A, B \cup J)$ almost surely. This is equivalent to identifying a sufficiently large region for which either $I$ is exited or the extended Streett pair $(A, B \cup J)$ is satisfied. 
    As a result, in this example, the probability of exiting $I$ serves to upper-bound the probability of violating the original reactivity condition. \lipicsEnd
\end{example}

Our technical results formalise the intuition by which adding absorbing regions $J$ and invariant regions $I$ provides completeness while retaining soundness, as we illustrated in \cref{ex:orey-incomplete,ex:orey-no-quant}. 
Crucially, we show that this approach generalises to every time-homogeneous Markov chain on a general state space. 

\begin{figure}[t]
\begin{tikzpicture}
[x=0em,y=0em,
box/.style={draw=black,minimum width=0.5*\linewidth-2pt,rounded corners=0.2em, minimum height=2em},
header/.style={minimum width=12em}]
\node (origin) at (0,0) {};

\def\dist{4pt}

\node[box, minimum width=\linewidth] (meta) at (0,0) {\cref{thm:meta}$: \Prob((I^\omega)^{\sf c} \cup L)=1 \implies \exists J \mathpunct. \text{Eqs.}\, \eqref{eqn:semantic-1-quant}\land \eqref{eqn:semantic-2-quant}$ };
\node[box, below=2*\dist of meta, minimum width=\linewidth] (soundness) {\cref{thm:absorbing-regions-decomposition} (Soundness): $\forall I,J \mathpunct.\text{ Eqs.}\, \eqref{eqn:semantic-1-quant}\land \eqref{eqn:semantic-2-quant} \Rightarrow \Prob(L) \geq \Prob(I^\omega) $};
\node[box, above=\dist of meta.north west, anchor=south west] (quantprobbound) {\cref{lem:cav25-derivative}: $\exists I \mathpunct. \Prob(I^\omega)\geq \Prob(L)-\varepsilon$};
\node[box, above=\dist of meta.north east, anchor=south east] (probbound) {\cref{thm:inv-exists-as}: $\exists I \mathpunct. \Prob(I^\omega)=\Prob(L)=1$};
\node[box, above=\dist of quantprobbound.north west, anchor=south west] (epscompleteness) {\cref{thm:eps-completeness} ($\varepsilon$-Completeness)};
\node[box, above=\dist of probbound.north west, anchor=south west,minimum width=0.5*\linewidth-0.66*\dist] (completeness) {\cref{thm:completeness} (Completeness)};
\node[header, above=\dist of completeness] (as) {\textbf{Almost-Sure}};
\node[header,above=\dist of epscompleteness] (quant) {\textbf{Quantitative}};

\draw[dashed] ([yshift=0.5em]soundness.north west) -- ([yshift=0.5em]soundness.north east);
\end{tikzpicture}
\caption{Structure of our Technical Results}
\label{fig:structure-completeness}
\end{figure}

Our arguments build on the scaffolding illustrated in \cref{fig:structure-completeness}. 
We present our decomposition in full generality and show its soundness in \cref{thm:absorbing-regions-decomposition}. We then prove in \cref{thm:meta} that every Markov chain that either violates an invariant $I$ or satisfies a reactivity property necessarily admits appropriate absorbing regions $J$. 
This provides the basis for establishing the completeness of our decomposition from both the quantitative and the qualitative (almost-sure) perspectives. In the general quantitative setting, we show that every Markov chain admits a sufficiently strong invariant $I$, whose probability of invariance lower-bounds the probability of the reactivity property with any arbitrary $\varepsilon > 0$ gap; this is shown in \cref{lem:cav25-derivative} which in turn establishes our $\varepsilon$-completeness result of \cref{thm:eps-completeness}. In the special case that the reactivity property is satisfied almost surely, we show that this gap is zero, which in turn establishes our completeness result for almost-sure reactivity of \cref{thm:completeness}.

\begin{theorem}[Absorbing-Region Decomposition]\label{thm:absorbing-regions-decomposition}
    Let $(A_1,\allowbreak B_1) \in \Sigma^2, \dots, (A_k,B_k) \in \Sigma^2$ be a Streett acceptance condition. 
    Suppose there exists a region $I \in \Sigma$ and $k$ regions $J_1 \in \Sigma, \dots, J_k \in \Sigma$ 
    such that, for every $i = 1, \ldots, k$, $J_i \subseteq I \smallsetminus A_i$ and
    \begin{align}
        &\sup_{s \in J_i } \Prob_{\delta_s}\left(~{\sigma}_{A_i} < \infty ~\right) < 1,\label{eqn:semantic-1-quant}\\
        &\inf_{s \in I}  \Prob_{\delta_s}\left(~\sigma_{B_i \cup J_i \cup I^{\sf c}} < \infty ~\right) = 1.\label{eqn:semantic-2-quant}
    \end{align}
    
\noindent Then
    \begin{equation}\label{lb-streett}
        \Prob_\mu \left( \bigcap_{i = 1}^k \Fin (A_i) \cup \Inf (B_i)  \right) \geq \Prob_\mu(I^\omega).
    \end{equation}

\end{theorem}

\begin{proof}
Let $i$ be arbitrary among $\{1, \ldots, k\}$.
We apply Orey's Theorem (\cref{thm:orey}) to a consequence of \cref{eqn:semantic-2-quant} obtained by restricting the infimum to a smaller set $A_i \cap I \subseteq I$:
\begin{align}\label{eqn:inf-AI}
    \inf_{s \in A_i \cap I}
    \Prob_{\delta_s}\left( 
    ~\sigma_{B_i \cup J_i \cup I^{\sf c}} < \infty~
    \right)
    \geq 
    \inf_{s \in I}
    \Prob_{\delta_s}\left( 
    ~\sigma_{B_i \cup J_i \cup I^{\sf c}} < \infty~
    \right)
    = 1,
\end{align}
and Orey's Theorem applied to \eqref{eqn:inf-AI} yields:
\begin{align}\label{eqn:imm-orey}
    \Prob_\mu\left( 
    \Fin(A_i \cap I) \cup \Inf(B_i \cup J_i \cup I^{\sf c})
    \right) = 1.
\end{align}
We now establish the inclusion
\begin{align}\label{eqn:inclusion-1}
    \Fin(A_i \cap I) \subseteq 
    (I^\omega)^{\sf c}
    \cup 
    (\Fin(A_i) \cap I^\omega) 
\end{align}
so that we can replace $\Fin(A_i \cap I)$ in \cref{eqn:imm-orey} with $ (I^\omega)^{\sf c} \cup (\Fin(A_i) \cap I^\omega)$ to obtain \cref{eqn:replaced}.
To this end, we first note that $\Fin(A_i \cap I) \cap I^\omega = \Fin(A_i) \cap I^\omega$, because 
for any trajectory $(s_0,s_1,\cdots, s_n, \cdots) \in I^\omega$, and any time $n \in \Nat$, we have $s_n \in A_i \cap I \Leftrightarrow s_n \in A_i$, of which \cref{eqn:inclusion-1} is a consequence. We thus obtain:
\begin{align}\label{eqn:replaced}
    \Prob_\mu( 
    (I^\omega)^{\sf c}
    \cup 
    (\Fin(A_i) \cap I^\omega) 
    \cup 
    \Inf(B_i \cup J_i \cup I^{\sf c}
    )) = 1
\end{align}
Using the fact that $\Inf(B_i \cup J_i \cup I^{\sf c}) = \Inf(B_i) \cup \Inf(J_i) \cup \Inf(I^{\sf c})$, 
and that any trajectory performing infinitely many visits to $I^{\sf c}$ must necessarily leave $I$, namely:
\begin{align}
    \Inf(I^{\sf c}) \subseteq (I^\omega)^{\sf c},
\end{align} 
we rewrite
\cref{eqn:replaced} into \cref{eqn:replaced-3}:
\begin{align}\label{eqn:replaced-3}
    \Prob_\mu( 
    (I^\omega)^{\sf c}
    \cup 
    (\Fin(A_i) \cap I^\omega) 
    \cup 
    \Inf(B_i) \cup \Inf(J_i)
    ) = 1.
\end{align}
By set-theoretic reasoning, we replace $(I^\omega)^{\sf c} \cup (\Fin(A_i) \cap I^\omega)$ in \cref{eqn:replaced-3} by $(I^\omega)^{\sf c} \cup \Fin(A_i)$, since
\begin{align}
    &\overbrace{
    (I^\omega)^{\sf c} \cup 
    (\Fin(A_i) \cap (I^\omega)^{\sf c})}^{ (I^\omega)^{\sf c}}
    \cup 
    (\Fin(A_i) \cap I^\omega)
    =
    (I^\omega)^{\sf c}
    \cup  
\Fin(A_i) \cap 
\overbrace{( (I^\omega)^{\sf c} \cup I^\omega )}^{\Omega}.\label{eqn:set-2}
\end{align}
Therefore we obtain:
\begin{align}\label{eqn:replaced-4}
    \Prob_\mu( 
    (I^\omega)^{\sf c}
    \cup 
    \Fin(A_i)
    \cup 
    \Inf(B_i) \cup \Inf(J_i)
    ) = 1.
\end{align}
Our next objective is to remove $\Inf(J_i)$ from \cref{eqn:replaced-4}. We first note that by \cref{eqn:semantic-1-quant} that there exists $\gamma < 1$ for which
\begin{align}\label{eqn:expanded-supremum}
\forall s \in J_i \mathpunct. 
    \Prob_{\delta_s}(~\sigma_{A_i} < \infty~)
    \leq \gamma.
\end{align}
Therefore, in any trajectory that visits $J_i$ infinitely many times, the probability of returning to $A_i$ must drop below $\gamma$ infinitely often. Formally, using the definition of $\Inf(J_i)$ in \cref{eqn:inf-definition} and \cref{eqn:expanded-supremum} we derive:
\begin{align}
    \chain \in \Inf(J_i) 
    \iff 
    &\forall n ~\exists m\geq n \mathpunct.
    \mc{m} \in J_i\\
    \implies 
    &\forall n ~\exists m\geq n \mathpunct.
    \Prob_{\delta_{\mc{m}}}
    (~\sigma_{A_i} < \infty~)
    \leq \gamma.
    \label{eqn:io-leq-gamma}
\end{align}
By \cref{lem:technical}, the limit 
$
    \lim_{n\to\infty }
    \Prob_{\delta_{\mc{n}}}(~\sigma_{A_i} < \infty~)    
$
exists and takes value either zero or one, $\Prob_{\mu}$-almost surely. We note that \eqref{eqn:io-leq-gamma} implies that the limit cannot be one, meaning that the limit must be equal to zero, $\Prob_\mu$-almost surely. Summarising, we conclude that there exists a $\Prob_\mu$-negligible set $N_3 \in \mathcal{F}$ for which
\begin{align}\label{eqn:lim-is-zero-mu-as}
    \chain \in \Inf(J_i) \smallsetminus N_3
    &\implies 
    \lim_{n \to \infty}
    \Prob_{\delta_{\mc{n}}}(
    ~\sigma_{A_i} < \infty~
    ) = 0.
\end{align}
Finally, since (by \cref{lem:technical}) $\lim_{n \to \infty}
    \Prob_{\delta_{\mc{n}}}(~
    \sigma_{A_i} < \infty~
    ) = 0$ and $\Fin(A_i)$ differ by a $\Prob_\mu$-negligible set, we conclude for some $\Prob_\mu$-negligible set $N_4$, that:
\begin{align}\label{eqn:InfJ-subseteq-FinA}
    \Inf(J_i) \smallsetminus N_4
    \subseteq \Fin(A_i).
\end{align}
This means we can replace $\Fin(A_i) \cup \Inf(J_i)$ in \cref{eqn:replaced-4} by $\Fin(A_i)$ to obtain:
\begin{align}
    \Prob_\mu( 
    (I^\omega)^{\sf c}
    \cup 
    \Fin(A_i)
    \cup 
    \Inf(B_i)
    ) = 1.
\end{align}
Since this holds for arbitrary $i = 1,\ldots, k$, we conclude the result of the theorem by noting that the intersection of finitely many probability 1 events has probability 1:
\begin{align}
    &\Prob_\mu\left( 
    (I^\omega)^{\sf c}
    \cup 
    \bigcap_{i = 1}^{k}
    \Fin(A_i)
    \cup 
    \Inf(B_i)
    \right) = 1,
\end{align}
and by applying a union bound we obtain 
\begin{align}
    \Prob_\mu\left(
    \bigcap_{i = 1}^{k}
    \Fin(A_i)
    \cup 
    \Inf(B_i)
    \right) 
    \geq 1 - \Prob_{\mu}( (I^\omega)^{\sf c} ) = \Prob_\mu(I^\omega)
\end{align}
which is the desired result.
\end{proof}

\begin{figure}[t]
    \centering
        \begin{tikzpicture}[minimum size=7mm, node distance=14mm]
        \node[draw, circle, initial, initial text=$\mu$] (1) {$s_0$};
        
        \node[draw, circle, right of=1,label={[label distance=-5pt]above:{$\{a\}$}}] (2) {$s_1$}; 

        \node[draw, circle, right of=2] (6) {$s_6$};
        
        \node[draw,circle,right of=6,label={[label distance=-5pt]above:{$\{b\}$}}] (3) {$s_2$};
        
        \node[draw, circle, below of=2,label={[label distance=-5pt]above:{$\{a\}$}}] (0) {$s_3$};
        \node[draw, circle, right of=0] (m1) {$s_4$};

        \node[draw, circle, above of=2,label={[label distance=-5pt]above:{$\{a\}$}}] (5) {$s_5$};

        \draw (1) edge[->] node[below, yshift=1mm] {$\frac{1}{3}$} (2)
        (2) edge[->,bend left] node[above] {} (6)
        (6) edge[->,bend left] node[below,yshift=1mm] {$\frac{1}{2}$} (2)
        (6) edge[->] node[above,yshift=-1mm] { $\frac{1}{2}$ } (3)
        (3) edge[->, loop right] (3)
        (1) edge[->] node[below left,xshift=2mm, yshift=2mm] {$\frac{1}{3}$} (0)
        (0) edge[->] node[above] {} (m1)
        (1) edge[->] node[above left,xshift=3mm,yshift=0.5mm] {$\frac{1}{3}$} (5)
        (5) edge[->, loop right] (5)
        (m1) edge[->, loop right] (m1);

        \node (J) [draw=orange, fit= (m1) (6), inner sep=0.2cm, dashed, thick, fill=orange!20, fill opacity=0.2] {};
        \node [xshift=3.1ex,yshift=-4.0ex, orange] at (J.east) {$J$};

        \node (I) [draw=yellow!10!green, fit= (m1) (0) (1) (2) (3), inner sep=0.5cm, dashed, thick, fill=green!20, fill opacity=0.2] {};
        \node [xshift=2.0ex, yellow!10!green] at (I.east) {$I$};

        \node (myTableNode) [right of=3, yshift=0cm,xshift=3.5cm] {
        {\setlength{\tabcolsep}{1pt}
        \begin{tabular}{l l}
            $s_0 (s_5)^\omega $ 
            &$\in (I^\omega)^{\sf c}$  
            \\
            $s_0 s_3 (s_4)^\omega$ 
            &$\in I^\omega \cap \Fin(A) \cap \Fin(B)$\\
            $s_0 s_1 (s_6 s_1)^\omega $ &$\in I^\omega 
            \cap \Inf(A) \cap \Fin(B)$\\
            $s_0 s_1 (s_6 s_1)^n s_6 (s_2)^\omega $
            &$\in I^\omega \cap \Fin(A) \cap \Inf(B)$\\
            $\varnothing$ &= $I^\omega \cap \Inf(A) \cap \Inf(B)$
        \end{tabular}
        }
    };
    
    \end{tikzpicture}
    \caption{Intuition for Absorbing-Region Decomposition}
    \label{fig:soundness_absorbing_region_intuition}
\end{figure}
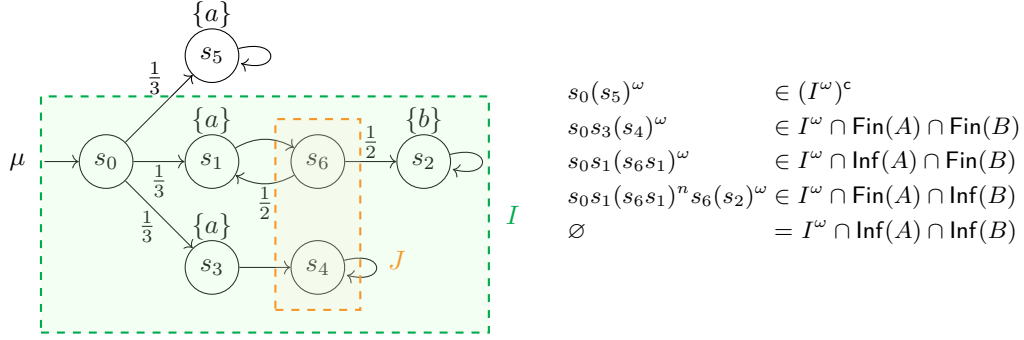
\begin{example}[Intuition for \cref{thm:absorbing-regions-decomposition}]
Consider \cref{fig:soundness_absorbing_region_intuition}, which displays a Markov chain that satisfies the requirements of \cref{thm:absorbing-regions-decomposition} and therefore satisfies the Streett condition with probability $\frac{2}{3}$.
In the previous examples, we could not return from the set $J$ to $A$ at all. However, in this example, we can go from the state $s_6$ to $s_1$ with probability $\frac{1}{2}$.

We observe that any path that visits the set $A$ infinitely often must either exit the set $I$, as before, or visit the set $B \cup J$ infinitely often. If $B$ is visited infinitely often, we satisfy Streett. On the other hand, we see that for the choice of $J = \{ s_4, s_6 \}$ that the probability of returning to $A$ from it is no greater than $\frac{1}{2}$, which means that the probability of returning to $A$ from it infinitely often is zero.
In our example, this corresponds to the fact that, in order to reach $A$ infinitely often, the Markov chain must either go to $s_5$, which is outside the set $I$, or take the transition from $s_6$ to $s_1$ infinitely often, the probability of which is zero.\lipicsEnd
\end{example}

\cref{thm:absorbing-regions-decomposition} establishes that if a Markov chain admits an invariant and absorbing regions then it satisfies the reactivity property with probability at least $\Prob_\mu(I^\omega)$. We next show, in \cref{thm:meta}, that given a sufficiently strong invariant region (namely, a region $I \in \Sigma$ satisfying \eqref{eqn:streett-or-ic}), there exist suitable absorbing regions, such that the invariant and the absorbing regions jointly satisfy requirements \eqref{eqn:semantic-1-quant} and \eqref{eqn:semantic-2-quant}.

\begin{lemma}\label{thm:meta}
Let $(A_1,B_1) \in \Sigma^2, \dots, (A_k,B_k) \in \Sigma^2$ be a Streett acceptance condition.
Let $I \in \Sigma$ be a region for which
\begin{align}\label{eqn:streett-or-ic}
    \forall s \in S \mathpunct. 
    \Prob_{\delta_s}\left( 
    (I^\omega)^{\sf c} \cup \bigcap_{i = 1}^{k}
    \Fin(A_i) \cup \Inf(B_i)
    \right) = 1.
\end{align}
Then, there exist regions $J_i \in \Sigma$ with $J_i \subseteq I \smallsetminus A_i$ for $i = 1, \ldots, k$ such that \cref{eqn:semantic-1-quant,eqn:semantic-2-quant} hold.
\end{lemma}
\begin{proof}
For every $i \in \{1 , \ldots, k \}$ we define the region
\begin{align}\label{eqn:defJ-completeness}
    J_i = \left\{ s \in I \smallsetminus A_i \colon 
    \Prob_{\delta_s}(~\sigma_{A_i} < \infty~) \leq 0.5
    \right\},
\end{align}
with $0.5$ being an arbitrarily chosen constant in $(0, 1)$.
The region $J_i$ of \eqref{eqn:defJ-completeness} satisfies \eqref{eqn:semantic-1-quant} by construction, since $\sup_{s \in J_i}\Prob_{\delta_s}(~\sigma_{A_i} < \infty~) \leq 0.5$.

\noindent The premise \eqref{eqn:streett-or-ic} is equivalent to
\begin{align}\label{eqn:with-Iw}
    \forall s \in S
    \mathpunct. 
    \forall i  = 1, \ldots, k
    \mathpunct. 
    \Prob_{\delta_s}
    \left( 
    (I^\omega)^{\sf c}
    \cup \Fin(A_i) \cup \Inf(B_i)
    \right) = 1,
\end{align}
and by \cref{eqn:set-2} we replace
$(I^\omega)^{\sf c} \cup \Fin(A_i)$ by $(I^\omega)^{\sf c} \cup (\Fin(A_i) \cap I^\omega)$ in \cref{eqn:with-Iw} to obtain:
\begin{align}\label{eqn:replaced-IcFaIw-IcFa-NEW}
    \forall s \in S \mathpunct. 
    \forall i = 1, \ldots, k \mathpunct. 
    \Prob_{\delta_s}(
    (I^\omega)^{\sf c} \cup 
    (I^\omega \cap \Fin(A_i))
     \cup \Inf(B_i)
    ) = 1.
\end{align}
We now show that $I^\omega \cap \Fin(A_i)$ in \cref{eqn:replaced-IcFaIw-IcFa-NEW} may be replaced by $\Inf(J_i)$ to obtain \cref{eqn:replaced-IcFaIw-IcFa2-NEW}.
To this end, we invoke \cref{lem:technical}, which says that 
the events $\Fin(A_i)$ and $\{ 
\lim_{n \to \infty} 
\Prob_{\delta_{\mc{n}}}(~\sigma_{A_i} < \infty~) = 0
\}$ are $\Prob_{\delta_s}$-equivalent for every $s \in S$. This means that for every $s \in S$ there exists a $\Prob_{\delta_s}$-negligible event $N_s \in \mathcal{F}$ for which 
    $\Fin(A_i) \smallsetminus N_s \subseteq 
    \{
    \lim_{n \to \infty} 
    \Prob_{\delta_{\mc{n}}}(~\sigma_{A_i} < \infty ~)  = 0
    \}
    $.
We extend this to the following inclusions among events, which hold for every $s \in S$:
\begin{align}
    I^\omega \cap \Fin(A_i) \smallsetminus N_s
    &\subseteq 
    I^\omega \cap 
    \Fin(A_i) \cap 
    \left\{ 
    \lim_{n \to \infty} 
    \Prob_{\delta_{\mc{n}}}
    (~\sigma_{A_i} < \infty~) = 0
    \right\} 
    \label{eqn:first-in-inclusions}
    \\
    &=
    I^\omega \cap 
    \Fin(A_i) \cap 
    \left\{ 
    \forall \varepsilon_1 > 0 ~
    \exists n ~  
    \forall m \geq n \mathpunct.
    \Prob_{\delta_{\mc{m}}}(~\sigma_{A_i} < \infty ~) 
    \leq \varepsilon_1 
    \right\}
    \\
    &\subseteq 
    I^\omega \cap 
    \Fin(A_i) \cap 
    \left\{ 
    \exists n ~
    \forall m \geq n \mathpunct. 
    ~\Prob_{\delta_{\mc{m}}}(~\sigma_{A_i} < \infty~) \leq 0.5
    \right\}.
    \label{eqn:event-Iw-FinA-FG-lt-0.5-NEW}
\end{align}
We next observe that any trajectory that always remains within $I$ and visits $A_i$ finitely often must eventually always remain inside $I \smallsetminus A_i$:
\begin{align}
    \chain \in I^\omega \cap \Fin(A_i) &\iff
    \forall n \in \Nat \mathpunct. 
    \Phi_n \in I
    \land 
    \exists n ~ \forall m \geq n \mathpunct. 
    \Phi_m \notin A_i\\
    &\implies \exists n ~ \forall m \geq n \mathpunct. \Phi_m \in I \smallsetminus A_i.
    \label{eqn:ltl-reasoning}
\end{align}
Equations \eqref{eqn:first-in-inclusions} to \eqref{eqn:ltl-reasoning}, combined with the definition of $J_i$ in \eqref{eqn:defJ-completeness}, shows that $I^\omega \cap \Fin(A_i) \smallsetminus N_s$ is a subset of trajectories that visit $J_i$ infinitely often, namely:
\begin{align}
    I^\omega \cap \Fin(A_i) 
    \smallsetminus N_s 
    &\subseteq 
    \left\{ 
    \exists n ~\forall m\geq n\mathpunct. 
    \mc{m}\in I \smallsetminus A_i 
    \land \Prob_{\delta_{\mc{m}}}(
    ~\sigma_{A_i} < \infty~) \leq 0.5 
    \right\}
    \label{eq:start-Iw-FinA-mN}
    \\
    &= 
    \left\{ 
    \exists n ~ \forall m \geq n \mathpunct. \mc{m} \in J_i
    \right\}
    = \Fin(J_i^{\sf c}) \subseteq \Inf(J_i).
    \label{eq:end-InfJc}
\end{align}
Using inclusions \eqref{eq:start-Iw-FinA-mN} and \eqref{eq:end-InfJc} to replace the event $I^\omega \cap \Fin(A_i) $ in \cref{eqn:replaced-IcFaIw-IcFa-NEW} by the larger event $\Inf(J_i)$ we obtain:
\begin{align}\label{eqn:replaced-IcFaIw-IcFa2-NEW}
    \forall s \in S \mathpunct. 
    \forall i = 1, \ldots, k \mathpunct.
    \Prob_{\delta_s}(
    (I^\omega)^{\sf c} \cup 
    \Inf(J_i)
     \cup \Inf(B_i)
    ) = 1.
\end{align}
We next seek to show that \cref{eqn:replaced-IcFaIw-IcFa2-NEW} implies \cref{eqn:semantic-2-quant}, by analysing its consequences upon the return times to the regions $B_i$, $J_i$, and $I^{\sf c}$. We observe that to visit a region $X \in \Sigma$ infinitely often a trajectory must return to it at least once:
\begin{align}
    \chain \in \Inf(X) \iff
    \left[ \forall n ~ \exists m \geq n \mathpunct. \mc{m} \in X 
    \right]
    \implies 
    \left[ \exists m \geq 1 \mathpunct. \mc{m} \in X \right]
    \iff \sigma_X < \infty, 
\end{align} 
and therefore,
\begin{align}
    &\Inf(J_i) \subseteq \{ \sigma_{J_i} < \infty \},
    && 
    \Inf(B_i) \subseteq \{ \sigma_{B_i} < \infty \}.
    \label{eqn:IR-B}
\end{align}
Furthermore, the trajectories that satisfy $(I^\omega)^{\sf c}$ are precisely those for which $\tau_{I^{\sf c}} < \infty$:
\begin{align}\label{eqn:Iwc-iff-tIc-lt-inf}
    \chain \in (I^\omega)^{\sf c}
    \iff 
    \neg [\forall n
    \mathpunct. \mc{n} \in I]
    \iff 
    [\exists n \geq 0 \mathpunct. 
    \mc{n} \in I^{\sf c}]
    \iff
    \tau_{I^{\sf c}} < \infty.
\end{align}
Using  
\eqref{eqn:IR-B}, and 
\eqref{eqn:Iwc-iff-tIc-lt-inf} we see that \eqref{eqn:replaced-IcFaIw-IcFa2-NEW} implies 
\begin{align}\label{eqn:tau-sigma-sigma-NEW}
    \forall s \in S \mathpunct. 
    \forall i = 1, \ldots, k \mathpunct.
    \Prob_{\delta_s}\left(~
    \tau_{I^{\sf c}} < \infty \lor \sigma_{J_i} < \infty \lor \sigma_{B_i} < \infty 
    ~\right) = 1.
\end{align}
Next, we observe that for every $s \in I$ that the events
    $\{ \tau_{I^{\sf c}} < \infty \} $ 
    and $\{ \sigma_{I^{\sf c}} < \infty \}$ are $\Prob_{\delta_s}$-equivalent, because for any trajectory whose initial state $s$ belongs to $I$, the hitting time to $I^{\sf c}$ coincides with the return time to $I^{\sf c}$:
\begin{align} 
\mc{0} \in I \implies
(
\underbrace{\exists n \geq 0 \mathpunct. \mc{n} \in I^{\sf c}}_{
\tau_{I^{\sf c}} < \infty
}
\iff
\underbrace{\exists n \geq 1 \mathpunct. \mc{n} \in I^{\sf c}}_{\sigma_{I^{\sf c}} < \infty}
)
\end{align}
holds pointwise. This allows us to rewrite \cref{eqn:tau-sigma-sigma-NEW} into 
\begin{align}\label{eqn:sigma-sigma-sigma}
    \forall s \in I \mathpunct. 
    \forall i = 1, \ldots, k \mathpunct.
    \Prob_{\delta_s}\left(
    ~
    \sigma_{I^{\sf c}} < \infty \lor \sigma_{J_i} < \infty \lor \sigma_{B_i} < \infty
    ~
    \right) = 1.
\end{align}
by restricting the quantifier to $s \in I$.
We observe that
\begin{align}
    \sigma_{I^{\sf c}} < \infty 
    \lor \sigma_{J_i} < \infty 
    \lor \sigma_{B_i} < \infty 
    \iff 
    \sigma_{B_i \cup J_i \cup I^{\sf c}} < \infty
\end{align}
which allows us to rewrite \cref{eqn:sigma-sigma-sigma} into
\begin{align}\label{eqn:sigma-once}
\forall s \in I \mathpunct. 
    \forall i = 1, \ldots, k \mathpunct. 
    \Prob_{\delta_s}\left(
    ~
    \sigma_{B_i \cup J_i \cup I^{\sf c}} < \infty
    ~
    \right) = 1, 
\end{align}
from which \cref{eqn:semantic-2-quant} follows.
\end{proof}

The premise of \cref{thm:meta} is that the invariant $I$ is sufficiently
strong, in the sense of condition \eqref{eqn:streett-or-ic}, which requires that almost every trajectory
that remains forever within $I$ also satisfies the reactivity property
$\bigcap_{i = 1}^k \Fin(A_i) \cup \Inf(B_i)$. Equivalently, the
probability of the reactivity property conditional upon $I^\omega$ must equal
$1$ (cf.\ \cite[Theorem 2]{DBLP:conf/fm/ChatterjeeGGKZ24},
\cite[Theorem 6]{AbateGR25}). Given any such invariant, \cref{thm:meta}
guarantees the existence of absorbing regions that, together with $I$,
satisfy \cref{eqn:semantic-1-quant,eqn:semantic-2-quant}.

\begin{example}[Necessity of Quantitative Safety]
In the quantitative setting, where the reactivity property fails to hold
almost surely, any sufficiently strong invariant must be exited with
positive probability. For instance, in
\cref{fig:soundness_absorbing_region_intuition}, no invariant satisfying 
\cref{eqn:streett-or-ic} may contain $s_5$, since otherwise the trajectory
$s_0 (s_5)^\omega$ would remain within $I$ while failing the reactivity
property. Since the Markov chain reaches $s_5$ with probability $1/3$, any
such $I$ must be exited with probability at least $1/3$. The joint
conditions on the absorbing regions and the invariant accommodate this
precisely, by permitting trajectories to exit $I$ in 
\cref{eqn:semantic-2-quant}.
\end{example}

It remains to show that a sufficiently strong invariant always exists. To
this end, we rely on the shift-invariance of reactivity properties. In
\cref{lem:cav25-derivative}, we show that there always exists a region
whose probability of invariance can be made arbitrarily close to that of
the shift-invariant property in question.

\begin{lemma}\label{lem:cav25-derivative}
Suppose $L \in \mathcal{F}$ is shift-invariant. Then, for every arbitrary $\varepsilon > 0$ there exists a region $I \in \Sigma$ such that
\begin{align}
    &\Prob_{\mu}(I^\omega) \geq 
    \Prob_{\mu}(L) - \varepsilon,\label{eqn:invar-eps-exist}\\
    &\forall s \in S \mathpunct. \Prob_{\delta_s} ((I^\omega)^{\sf c} \cup L) = 1.\label{eqn:invar-conditional-exist}
\end{align}
\end{lemma}
\begin{proof}
    We consider the event $\{ \inf_{n \in \Nat} \Prob_{\delta_{\mc{n}}}(L)  >  0 \}$ that the satisfaction probability of $L$, when evaluated along a trajectory, is uniformly bounded away from zero. We note that if, for some $k \in \Nat$, the satisfaction probability is always above $\frac{1}{k+1}$, then it is uniformly bounded away from zero. That is, we note 
    \begin{align}
        \left\{ \inf_{n \in \Nat} \Prob_{\delta_{\mc{n}}}(L) \geq \frac{1}{k+1} \right\}\subseteq \left\{ \inf_{n \in \Nat} \Prob_{\delta_{\mc{n}}}(L) > 0 \right\}
    \end{align} for every $k \in \Nat$, and that
    \begin{align}\label{eqn:mct-premise}
        \bigcup_{k \in \Nat} 
        \left\{ \inf_{n \in \Nat} \Prob_{\delta_{\mc{n}}}(L) \geq \frac{1}{k+1} \right\} = 
        \left\{ \inf_{n \in \Nat} \Prob_{\delta_{\mc{n}}}(L) > 0 \right\}.
    \end{align}
    We define, for each $k \in \Nat$, the region 
    \begin{align}
        I_k = \left\{ 
        s \in S \colon \Prob_{\delta_s}(L) \geq \frac{1}{k+1}
        \right\},
    \end{align}
    which has the property that the events $I_k^\omega$ and $\{ \inf_{n \in \Nat} 
    \Prob_{\delta_{\mc{n}}}( L ) \geq \frac{1}{k+1}
    \}$ are identical.
    
    By \cite[Theorem 3]{AbateGR25}, the events $\{ 
    \inf_{n \in \Nat}
    \Prob_{\delta_{\mc{n}}}(L)  > 0
    \}$ and $L$ are $\Prob_\xi$-equivalent for every initial distribution $\xi$. 
    Using this, and applying the Monotone Convergence Theorem (\cite[Theorem 2.59]{axler2019measure}) to \cref{eqn:mct-premise} yields $\lim_{k \to \infty}  \Prob_{\mu} ( 
    I_k^\omega) = \Prob_\mu(L)$, so there must exist $k_0 \in \Nat$ for which $\Prob_\mu(I_{k_0}^\omega) \geq \Prob_\mu(L) - \varepsilon$, i.e.\ \cref{eqn:invar-eps-exist}. Since $I_{k_0}^\omega \subseteq \{ 
    \inf_{n \in \Nat}
    \Prob_{\delta_{\mc{n}}}(L)  > 0
    \}$ and $\{ \inf_{n \in \Nat} \Prob_{\delta_{\mc{n}}}(L) > 0 \}$ is (for every $\xi$) $\Prob_\xi$-almost surely equivalent to $L$, we conclude that $(I_{k_0}^\omega)^{\sf c} \cup L$ holds $\Prob_{\delta_s}$-almost surely for every $s \in S$, i.e.\ \cref{eqn:invar-conditional-exist}.
\end{proof}

A consequence of \cref{lem:cav25-derivative} is that, for reactivity properties, a sufficiently strong invariant always exists, and can be chosen so that its probability of invariance lower-bounds the probability of reactivity up to an arbitrarily small gap (\cref{thm:eps-completeness}).

\begin{theorem}[{\protect $\varepsilon$-Completeness of Absorbing-Region Decomposition}]\label{thm:eps-completeness}
    Let $(A_1,B_1) \in \Sigma^2, \dots, (A_k,B_k) \in \Sigma^2$ be a Streett acceptance condition.
    Suppose
    \begin{equation}
        \Prob_\mu \left( \bigcap_{i = 1}^k \Fin (A_i) \cup \Inf (B_i)  \right) > 0.
    \end{equation}
    Then for every arbitrary $\varepsilon > 0$, there exists a region 
    $I \in \Sigma$, and regions $J_i \in \Sigma$ with $J_i \subseteq I \smallsetminus A_i$ for $i = 1, \dots, k$
    such that \cref{eqn:semantic-1-quant,eqn:semantic-2-quant} hold and 
    \begin{align}\label{eqn:lb-eps}
    &\Prob_{\mu}(I^\omega) \geq 
    \Prob_\mu \left( \bigcap_{i = 1}^k \Fin (A_i) \cup \Inf (B_i)  \right)
    - \varepsilon. 
    \end{align}
\end{theorem}

\begin{proof}
The reactivity property corresponding to the Streett acceptance condition is shift-invariant, and therefore by \cref{lem:cav25-derivative} there exists a region $I$ for which \cref{eqn:lb-eps} holds and for which
\begin{align}
    \forall s \in S \mathpunct. 
    \Prob_{\delta_s}\left( 
    (I^\omega)^{\sf c} \cup \bigcap_{i = 1}^{k}
    \Fin(A_i) \cup \Inf(B_i)
    \right) = 1.
\end{align}
Using this equality, we can invoke \cref{thm:meta} to show that there exist regions $J_i \in \Sigma$ with $J_i \subseteq I\smallsetminus A_i$ for $i = 1, \ldots, k$ such that \cref{eqn:semantic-1-quant,eqn:semantic-2-quant} hold, which completes the proof of the theorem.
\end{proof}

By contrast, in the almost-sure setting, there exists a region whose probability of invariance matches the probability of the shift-invariant property exactly (\cref{thm:inv-exists-as}).

\begin{lemma}\label{thm:inv-exists-as}
Suppose $L \in \mathcal{F}$ is shift-invariant 
and $\Prob_\mu(L) = 1$. Then there exists a region $I \in \Sigma$ such that
\begin{align}
    &\Prob_\mu(I^\omega) = 1, \label{eqn:concl-A}\\
    &\forall s \in S \mathpunct.
    \Prob_{\delta_s}( 
    (I^\omega)^{\sf c}
    \cup L
    ) = 1.\label{eqn:concl-B}
\end{align}
\end{lemma}
\begin{proof}

We define $I = \{s \in S \colon \Prob_{\delta_s}(L) = 1\}$. We prove that $I$ satisfies \cref{eqn:concl-A} (which also proves $I$ to be non-empty) and \cref{eqn:concl-B}.

The crux of our proof is that the stochastic process $n \mapsto \Prob_{\delta_{\mc{n}}}(L)$ is a martingale adapted to the filtered probability space $(\Omega, \mathcal{F}, \mathfrak{F}, \Prob_\xi)$ for every initial distribution $\xi$ on $(S, \Sigma)$, which in concrete terms means that
\begin{flalign} \label{eq:martingale_prop}
&\qquad 
\forall n \in \Nat \mathpunct.
\Expect_\xi \left(\Prob_{\delta_{\mc{n+1}}}(L) \mid \mathcal{F}_n\right) = \Prob_{\delta_{\mc{n}}}(L) &\Prob_\xi\text{-a.s.},
\qquad 
\end{flalign}
which we now prove. First, we let $\xi$ be an arbitrary initial distribution on $(S, \Sigma)$, and we recall that $L$ is shift-invariant, meaning that the following equality
\begin{align}
    \ind{\{\chain \in L\}} = \ind{\{\chain \circ \theta^n \in L\}},
\end{align}
holds pointwise for every $n \in \Nat$. Applying the conditional expectation operator $\Expect_\xi( \cdot \mid \mathcal{F}_n)$ to both sides:
\begin{flalign}\label{eqn:shift-expect}
& \qquad \forall n \in \Nat \mathpunct. \Expect_\xi \left(\ind{\{\chain \in L\}} \mid \mathcal{F}_n\right) = \Expect_\xi\left(\ind{\{\chain \circ \theta^n \in L\}} \mid \mathcal{F}_n\right) 
&
\Prob_\xi\text{-a.s.}
\qquad 
\end{flalign}
The Markov property (\cref{eqn:strong-markov-property}) states
\begin{flalign}\label{eq:strong_markov}
    &\qquad \Expect_\xi \left(\ind{\{\chain \circ \theta^n \in L\}} \mid \mathcal{F}_n\right) = \Prob_{\delta_{\mc{n}}}(L) &\Prob_\xi\text{-a.s.},
    \qquad 
\end{flalign}
and applying this to \cref{eqn:shift-expect} yields:
\begin{flalign}\label{eqn:intermediate}
&\qquad \forall n \in \Nat \mathpunct. \Expect_\xi \left(\ind{\{\chain \in L\}} \mid \mathcal{F}_n\right) = \Prob_{\delta_{ \mc{n} }}(L)
&\Prob_\xi\text{-a.s.}\qquad 
\end{flalign}
L{\'e}vy's 0-1 Law \cite[Theorem 5.5.8]{durrett2019probability} states
\begin{flalign} \label{eq:levy}
    &\qquad \lim_{n \to \infty} \Expect_\xi \left(\ind{\{\chain \in L\}} \mid \mathcal{F}_n\right) = \ind{\{\chain \in L\}} &\Prob_\xi\text{-a.s.}
    \qquad
\end{flalign}
Combining \cref{eqn:intermediate,eq:levy}, we have:
\begin{flalign}
&\qquad \ind{\{\chain \in L\}} = \lim_{n \to \infty} \Prob_{\delta_{\mc{n}}}(L)
&\Prob_\xi\text{-a.s.}\qquad
\end{flalign}
Therefore, for some $\Prob_\xi$-negligible event $N \in \mathcal{F}$, we have:
\begin{equation}\label{eq:subset_a}
\left\{ \lim_{n \to \infty} \Prob_{\delta_{\mc{n}}}(L) = 1 \right\} \smallsetminus N 
\subseteq L.
\end{equation}
We note that $I^\omega \smallsetminus N \subseteq \left\{ \lim_{n \to \infty} \Prob_{\delta_{\mc{n}}}(L) = 1 \right\}$ as:
\begin{align} \label{eq:subset_b}
I^\omega \smallsetminus N 
= \{ \forall n \in \Nat \mathpunct. \Prob_{\delta_{\mc{n}}}(L) = 1 \}
\smallsetminus N
\subseteq \left\{ \lim_{n \to \infty} \Prob_{\delta_{\mc{n}}}(L) = 1\right\} \smallsetminus N \subseteq L,
\end{align}
which proves that $\Prob_{\xi}((I^\omega)^{\sf c} \cup L) = 1$. 
Since $\xi$ was arbitrary, setting $\xi$ to be $\delta_s$ for every arbitrary initial state $s \in S$ concludes the proof of \cref{eqn:concl-B}.

\noindent We next prove \cref{eqn:concl-A}, noting that since 
\begin{align}
    \Prob_\mu(I^\omega) = 
    \Prob_\mu(\forall n \in \Nat \mathpunct. 
    \mc{n} \in I
    ) = 
    \Prob_\mu(
    \cap_{n \in\Nat} \{ \mc{n} \in I \}
    ), 
\end{align}
it suffices to establish
\begin{align}\label{eqn:all_in_inv}
    \forall n \in \Nat \mathpunct. 
    \Prob_{\mu}( 
    \mc{n} \in I
    ) = 1,
\end{align}
which we now prove by induction.

{\em (Base case)}
By \cref{eqn:measure-traj} we observe that:
\begin{align}
        \Prob_\mu(\mc{0} \in I) = \int_{s_0 \in I} \mu(\mathrm{d}s)
         = \mu(I)
\end{align}
and so it suffices to show $\mu(I) = 1$. 
By \cite[Proposition 3.1.3(ii)]{douc2018markov} and the premise $\Prob_\mu(L) = 1$, we obtain:
    \begin{align}
        \int_{s_0 \in S} \Prob_{\delta_{s_0}}(L) ~\mu(\mathrm{d}s_0)
        = \Prob_\mu(
        L
        ) = 1.
    \end{align}
    The term $\int_{s_0 \in S} \Prob_{\delta_{s_0}}(L) ~\mu(\mathrm{d}s_0)$ is the expectation of the $[0,1]$-valued random variable $s_0 \mapsto \Prob_{\delta_{s_0}}(L)$ with respect to the probability measure $\mu$; since its expected value is 1, the integrand must $\mu$-almost surely equal 1, that is $\mu(\{ s \in S \colon 
    \Prob_{\delta_s}(L) = 1
    \}) = 1$ 
    which is precisely $\mu(I) = 1$ (cf.\ \cite[Lemma 26(iv)]{pollard_2001}). This concludes the proof of the base case. 
    
\textit{(Inductive Step)}
    We prove that $\Prob_{\delta_{\mc{n+1}}} (L ) = 1$ holds $\Prob_{\mu}$-a.s.\ 
    under the assumption that 
    $\Prob_{\delta_{\mc{n}}}(L) = 1$ holds $\Prob_{\mu}$-a.s.
    We recall that the martingale property \eqref{eq:martingale_prop} of the process $n \mapsto \Prob_{\delta_{\mc{n}}}(L)$, instantiating the initial distribution $\xi$ to $\mu$ and applying the (unconditional) expectation operator $\Expect_\mu( \cdot) $ to both sides, yields: 
    \begin{align}\label{eqn:ind-conclusion}
    \Expect_{\mu}\left( 
    \Prob_{\delta_{\mc{n+1}}}(L)
    \right)
    \overset{{
    \scriptsize
    \text{\cite[Thm.\ 5.1.6]{durrett2019probability}}
    }
    }{=}
    \Expect_\mu\left( 
    \Expect_\mu\left( 
    \Prob_{\delta_{\mc{n+1}}
    }(L)
    \mid \mathcal{F}_n
    \right)
    \right) 
    \overset{
    {\scriptsize \text{Eq.~\eqref{eq:martingale_prop}
       }
    }
    }{=}
    \Expect_\mu\left( 
    \Prob_{\delta_{\mc{n}}}
    (L)
    \right)
    \overset{
        {
        \scriptsize
        \text{IH}
        }    
    }{=} 1 
    \end{align}
where we note that the expression $
\Expect_\mu( 
    \Prob_{\delta_{\mc{n}}}
    (L)
    )$ is equal to one since it is the expected value of a random variable that is $\Prob_\mu$-a.s.\ equal to one (by the inductive hypothesis). 
    \Cref{eqn:ind-conclusion} proves $\Expect_\mu( 
    \Prob_{\delta_{\mc{n+1}}}(L)
    ) = 1$, implying that $\Prob_{\delta_{\mc{n+1}}}(L) = 1$ holds $\Prob_\mu$-a.s, which concludes the inductive step, and the proof of the theorem.
\end{proof}

Combining \cref{thm:inv-exists-as} with \cref{thm:meta}, we conclude that whenever a Markov chain satisfies a reactivity property almost surely, there exists an almost-sure invariant and corresponding absorbing regions satisfying our decomposition (\cref{thm:completeness}).

\begin{theorem}[Completeness of Absorbing-Region Decomposition for Almost-Sure Reactivity]\label{thm:completeness}
Let $(A_1,B_1) \in \Sigma^2, \dots,\allowbreak (A_k,B_k) \in \Sigma^2$ be a Streett acceptance condition. Suppose that 
\begin{align}\label{eqn:streett-315}
    \Prob_\mu 
    \left( 
    \bigcap_{i = 1}^k \Fin(A_i) \cup \Inf(B_i)
    \right) = 1.
\end{align}
Then, there exists a region $I \in \Sigma$, and regions $J_i \in \Sigma$ with $J_i \subseteq I \smallsetminus A_i$ for $i = 1, \dots, k$ such that 
$\Prob_\mu(I^\omega) = 1$ and \cref{eqn:semantic-1-quant,eqn:semantic-2-quant} hold.
\end{theorem}
\begin{proof}
Using the fact that reactivity is shift-invariant (\cref{thm:reactivity-shift-invariant}), and that the reactivity property 
in \cref{eqn:streett-315}
holds $\Prob_\mu$-a.s., we invoke \cref{thm:inv-exists-as} to obtain an invariant $I \in \Sigma$ for which $\Prob_\mu(I^\omega) = 1$ and
\begin{equation}
    \forall s \in S \mathpunct.
    \Prob_{\delta_s}\left( 
    (I^\omega)^{\sf c}
    \cup \bigcap_{i = 1}^k \Fin(A_i) \cup \Inf(B_i)
    \right) = 1.
\end{equation}
But now we can simply apply \cref{thm:meta} to the latter equation to show the existence of regions $J_i \subseteq I \smallsetminus A_i $, for $i = 1, \ldots, k$, for which \cref{eqn:semantic-1-quant,eqn:semantic-2-quant} hold, which concludes the proof of the theorem.
\end{proof}


\section{Complete Supermartingale Certificates for Almost-Sure and Quantitative Reactivity}\label{sec:complete-certificates}

Our technical results of \cref{sec:decomp} form the basis for complete proof rules for probabilistic reactivity. They show that reactivity can be established by identifying suitable absorbing regions together with a global invariant, satisfying three conditions: (i) the qualitative (non–almost-sure) safety property of \cref{eqn:semantic-1-quant}, (ii) the almost-sure termination property of \cref{eqn:semantic-2-quant}, and (iii) the quantitative safety property used to lower-bound the right-hand side of \cref{lb-streett}, which in turn lower-bounds the probability of satisfying the reactivity property. \Cref{fig:rule_stack} summarises our framework, which lifts every complete proof rule for almost-sure termination into an $\varepsilon$-complete proof rule for quantitative reactivity.

In this section, we show that quantitative and qualitative safety admit sound and complete proof rules in \cref{cor:inv,lem:absorbing1}, respectively, over general state spaces. It then suffices to pick a sound and complete proof rule for almost-sure termination and substitute it into this framework to obtain a sound and complete proof rule for Streett conditions. Below, we carry out this instantiation with two pre-existing supermartingale proof rules for almost-sure termination (Proof Rules \#1 and \#2), and establish their soundness and completeness in the theorems that follow, with the proofs following from our results in \cref{sec:decomp}.

\begin{figure}[t]
    \centering
    \centering
\resizebox{\linewidth}{!}{%
\begin{tikzpicture}
[x=1em,y=1em,
box/.style={align=center,draw=black,rounded corners=0.2em, minimum height=4em},
label/.style={align=center, scale=0.75}
]
\def\diagwidth{420pt}                 
\def\diagheight{4em}
\def\dist{4pt}
\def\colw{(\diagwidth-2*\dist)/3}
\def\twocolw{2*\colw+\dist}
\node[box,minimum width=\diagwidth] (reactivity) at (0,0) {Quantitative Reactivity\\$\Prob_\mu(I^\omega) \leq \Prob_\mu(\bigcap_{i=1}^k \Fin(A_i) \cup \Inf(B_i)) \leq \Prob_\mu(I^\omega)+\varepsilon$};
\node[box, above=\dist of reactivity.north east, anchor=south east, minimum height=2*\diagheight+\dist, minimum width=\colw] (safety) {Quantitative Safety\\$\Prob_\mu(I^\omega)\overset{?}{=}p$};
\node[box, above=\dist of reactivity.north west, anchor=south west, minimum width=\twocolw] (asreactivity) {Almost-Sure Reactivity\\$\Prob_\mu((I^\omega)^{\sf c} \cup \bigcap_{i=1}^k \Fin(A_i) \cup \Inf(B_i)) = 1$};
\node[box, above=\dist of asreactivity.north west, anchor=south west, minimum height=\diagheight, minimum width=\colw] (avoidance) {\strut Qualitative Safety\\{\strut  $\sup\limits_{s\in J_i}\Prob_{\delta_s}({\sigma}_{A_i} < \infty)<1$}};
\node[box, above=\dist of asreactivity.north east, anchor=south east, minimum height=\diagheight, minimum width=\colw] (return) {\strut Almost-Sure Termination\\{\strut  $\inf\limits_{s\in I}\Prob_{\delta_s}(\sigma_{B_i \cup J_i \cup I^{\sf c}}<\infty)=1$}};
\node[box, dotted, fill=blue!7, thick, above=\dist of return.north west, anchor=south west, minimum width=\colw,minimum height=\diagheight+\dist] (byob) {Bring your own\\proof rule};
\node[box, above=\dist of avoidance.north west, anchor=south west, minimum width=\colw,minimum height=\diagheight+\dist] (supermart) {$\forall s\notin A_i\mathpunct. PV_i(s) \leq V_i(s)$ \\$\forall s\in A_i\mathpunct. V_i(s)\geq 1$\\$\exists \gamma_i > 0~ \forall s \in J_i\mathpunct. V_i(s)\leq 1-\gamma_i$};
\node[box,above=\dist of safety.north west, anchor=south west, minimum width=\colw,minimum height=\diagheight+\dist] (stochastic) {$\forall s \in I\mathpunct.PV_0(s)\leq V_0(s)$\\$\forall s \notin I \mathpunct.V_0(s)\geq 1$};
\node (certificates) [draw=blue, fit= (byob) (supermart) (stochastic), inner sep=0.45*\dist, dashed, thick] {};
\node (decompositions) [draw=red, fit= (avoidance) (return) (safety) (asreactivity), inner sep=0.45*\dist, dashed, thick] {};
\node[color=red,label,left=0em of decompositions.west,anchor=south,rotate=90] {Absorbing Region\\Decomposition};
\node[color=blue,label,left=0em of certificates.west,anchor=south,rotate=90] {Supermartingale\\Certificates};
\node[color=yellow!10!green,label,right=0em of reactivity.east,anchor=south,rotate=-90] {Section 2};
\node[color=red,label,right=0em of decompositions.east,anchor=south,rotate=-90] {Section 3};
\node[color=blue,label,right=0em of certificates.east,anchor=south,rotate=-90] {Section 4};
\end{tikzpicture}%
}
    \caption{Methodology Underlying our Supermartingale Certificates for Quantitative Reactivity}
    \label{fig:rule_stack}
\end{figure}
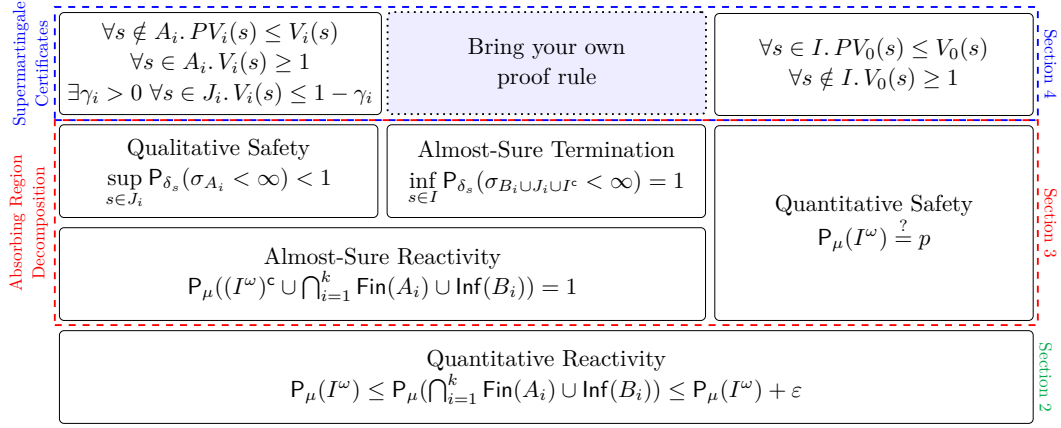

Supermartingale certificates provide proof rules for temporal properties of stochastic processes and are expressed as non-negative, real-valued measurable functions $V \colon S \to \Real_{\geq 0}$ over the state space $S$. Such certificates are typically required to satisfy conditions on their expected values with respect to the transition kernel $P$ and the initial distribution $\mu$ of a time-homogeneous Markov chain. 

Two standard tools that we will need in this setting are the post-expectation $(PV) \colon S \to \Real_{\geq 0}$ and the init-expectation $(\mu V) \in \Real_{\geq 0}$ of $V$, respectively defined as follows:
\begin{align}
    &PV(s) = \int_{u\in S} V(u) ~P(s, \d u), &&\mu V = \int_{s\in S} V(s)~\mu(\d s).
\end{align}
Intuitively, the post-expectation is just an expected value formula and gives the expected value of $V$ at the next state given that we are in a state $s$.
On the other hand, the init-expectation is the expected value of $V$ at the first state given an initial distribution $\mu$.

We instantiate the safety obligations of our framework using a standard result for quantitative safety which requires the identification of a value function to upper bound the probability that the stochastic process under consideration leaves a given region.

\begin{theorem}[{{\protect \cite[Corollary 4.4.7]{douc2018markov}}}]\label{cor:inv}
Let $X \in \Sigma$ be a region. 
\begin{enumerate}[i)]
    \item Suppose there exists a measurable function $V : S \to \Real_{\geq 0}$ for which the following conditions hold  
    \begin{align}
    &\forall s \in X \mathpunct. PV(s) \leq V(s)\label{eq:si_1},\\
    &\forall s \notin X \mathpunct. V(s) \geq 1\label{eq:si_2}.
\end{align}
    Then, $\Prob_\xi( X^\omega ) \geq 1 - \xi V$ for every initial distribution $\xi$ on $(S, \Sigma)$.

    \item Suppose $\xi$ is a probability distribution on $(S, \Sigma)$ for which $\Prob_{\xi}(X^\omega) \geq p$, for some $p \in [0, 1]$. Then, 
    there exists a measurable function $V : S \to \Real_{\geq 0}$ 
    for which \cref{eq:si_1,eq:si_2} hold, and for which $\xi V \leq 1 - p$.
\end{enumerate}
\end{theorem}

\Cref{cor:inv} provides a sound and complete proof rule for quantitative safety
with respect to an invariant $I$. Combined with \cref{lb-streett}, it yields a
lower bound on the probability of reactivity, provided that reactivity holds
almost surely within $I$, which in turn is ensured by
\cref{eqn:semantic-1-quant,eqn:semantic-2-quant}. It remains to give a proof
rule for the qualitative safety obligation of \cref{eqn:semantic-1-quant}.
To
this end, we apply \cref{cor:inv} to prove
\cref{lem:absorbing1}, which provides a sound and complete supermartingale certificate for
establishing that $J$ is an absorbing region with respect to $A$, that is,
that the probability of returning from $J$ to $A$ is uniformly bounded
below one.

\begin{theorem}\label{lem:absorbing1}
Suppose $A, J \in \Sigma$ are regions with $A$ disjoint from $J$.
Then, there exists a function $V \colon S \to \Real_{\geq 0}$ for which 
\begin{align}
    \forall s \notin A  &\mathpunct. PV(s) \leq V(s),\label{eqn:decrease2}\\
    \forall s \in A &\mathpunct. V(s) \geq 1,\label{eqn:safety2}\\
    \exists \gamma > 0 ~
    \forall s \in J &\mathpunct. 
    V(s) \leq 1 - \gamma,\label{eqn:supcondition2}
\end{align} 
if and only if $\sup_{s \in J} \Prob_{\delta_s}( ~\sigma_A < \infty~) < 1$.
\end{theorem}

\begin{proof}
$(\Rightarrow)$ Suppose $V \colon S \to \Real_{\geq 0}$ is a function for which \cref{eqn:decrease2,eqn:safety2,eqn:supcondition2} hold, then 
by \cref{cor:inv}(i) we conclude
\begin{align}\label{eqn:UB-exit2}
    \forall s \in S 
    \mathpunct.
    \Prob_{\delta_s}(~\tau_A < \infty~) 
    = 
    1 - \Prob_{\delta_s}((A^{\sf c})^\omega) \leq \delta_s V = V(s).
\end{align}
Since $A$ and $J$ are disjoint, we note that for any trajectory initialised in $J$, the first hitting time to $A$ and the first return time to $A$ coincide:
\begin{align}\label{eqn:equality-tau-sigma-outside-A}
    \forall s \in J \mathpunct. \Prob_{\delta_s} (~\sigma_A = \tau_A~) = 1.
\end{align}
This allows us to replace $\tau_A$ by $\sigma_A$ in \cref{eqn:UB-exit2} once we restrict the quantifier to the subset $J$ which is disjoint from $A$:
\begin{align}\label{eqn:using-sigma}
    \forall s \in J
    \mathpunct.
    \Prob_{\delta_s}(~\sigma_A < \infty ~) 
    \leq  V(s).
\end{align}
Applying \cref{eqn:supcondition2} to \cref{eqn:using-sigma} we obtain that the return probability to $A$ is upper bounded by $1 - \gamma$ when the initial state belongs to $J$:
\begin{align}\label{eqn:sup-expanded} 
    \forall s \in J \mathpunct. 
    \Prob_{\delta_s}(~\sigma_A < \infty~) \leq 1 - \gamma,
\end{align}
which implies $\sup_{s \in J} \Prob_{\delta_s}(~\sigma_A < \infty~) < 1$, concluding the proof of the forward direction. 

$(\Leftarrow)$ 
We define $V(s) = \Prob_{\delta_s}(~ \tau_A < \infty~) $ which by \cite[Corollary 4.4.7]{douc2018markov} satisfies conditions \eqref{eqn:decrease2} and \eqref{eqn:safety2}.
Moreover, we note that the premise $\sup_{s \in J} \Prob_{\delta_s}(~ \sigma_A < \infty ~) < 1$ is equivalent to the existence of a constant $\gamma > 0$ such that 
\begin{align}
    \forall s \in J \mathpunct. \Prob_{\delta_s}(~\sigma_A < \infty~) \leq 1 - \gamma,
\end{align}
which (recalling \cref{eqn:equality-tau-sigma-outside-A}) is precisely \cref{eqn:supcondition2} for our choice of $V$, which completes the reverse direction of the proof.
\end{proof}

\begin{figure}[h]
\centering
\begin{tikzpicture}[minimum size=7mm, node distance=14mm]
    \node[draw, circle] (1) {1};
    
    \node[draw, circle, above right of=1] (2) {2}; 
    \node[draw,circle,right of=2] (3) {3};
    \node[draw,circle,right of=3] (4) {4};
    \node[right of=4] (5) {};
    
    \node[draw, circle, below right of=1] (0) {0};
    \node[draw, circle, right of=0] (m1) {$-1$};
    \node[draw, circle, right of=m1] (m2) {$-2$};
    \node[right of=m2] (m3) {};
    \draw (1) edge[->, bend left] node[above] {$\frac{1}{2}$} (2)
    (2) edge[->, bend left] node[below,xshift=2mm,yshift=3mm] {$\frac{1}{2}$} (1)
    (2) edge[->,bend left] node[above] {$\frac{1}{2}$} (3)
    (3) edge[->,bend left] node[below] {$\frac{1}{2}$} (2)
    (3) edge[->,bend left] node[above] {$\frac{1}{2}$} (4)
    (4) edge[->,bend left] node[below] {$\frac{1}{2}$} (3)
    (4) edge[->, bend left, dashed] node[above] {$\frac{1}{2}$} (5)
    (5) edge[->, bend left, dashed] node[below] {$\frac{1}{2}$} (4);
    \def\pPlus{$\frac{1+\varepsilon}{2}$}
    \def\pMinus{$\frac{1-\varepsilon}{2}$}
    \draw (1) edge[->, bend left] node[above,xshift=3mm,yshift=-3.3mm] {$\frac{1}{2}$} (0);
    \draw (0) edge[->, bend left] node[below] {$\frac{1}{2}$} (1);
    \draw (0) edge[->, bend left] node[above] {$\frac{1}{2}$} (m1);
    \draw (m1) edge[->, bend left,yshift=0.5mm] node[below] {\pMinus} (0);
    \draw (m1) edge[->, bend left] node[above] {\pPlus} (m2);
    \draw (m2) edge[->, bend left] node[below] {\pMinus} (m1);
    \draw
    (m2) edge[->, bend left, dashed] node[above] {\pPlus} (m3)
    (m3) edge[->, bend left, dashed] node[below] {\pMinus} (m2);
    \node (J) [fit= (m1)(m2)(m3), inner sep=0.15cm, fill=orange!20, fill opacity=0.2] {};
    \draw[thick,dashed,orange] ($(m1)+(-5mm,5mm)$) -- ($(m3)+(5mm,5mm)$)
    ($(m1)+(-5mm,5mm)$) -- ($(m1)+(-5mm,-5mm)$) 
    ($(m1)+(-5mm,-5mm)$) -- ($(m3)+(5mm,-5mm)$);
    \node [xshift=-3.0ex, orange] at (J.east) {$J$};

    \node[anchor=west] (V0eq) at ([xshift=18mm]4.east)
        {$V_0(s) = 0$};
    \node[anchor=west] (V1eq) at (V0eq.west |- 0)
        {$V_1(s) = \begin{cases}
            \left(\frac{1-\varepsilon}{1+\varepsilon}\right)^{|s|} & \text{if } s < 0,\\
            1 & \text{if } s \geq 0
         \end{cases}$};

\end{tikzpicture}
\caption{Debt is an Absorbing Region}
\label{fig:casino-soundness}
\end{figure}
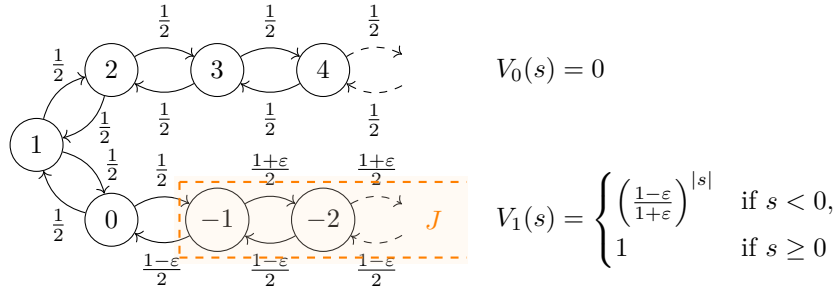

\begin{example}[Debt is an Absorbing Region]\label{eqn:debt-absorbing}
We illustrate our absorbing-region decomposition in the context of the
Lending Casino (\cref{fig:lending-casino}) for the property $\Fin(\textit{Solvency})$, which corresponds
to the Streett pair $(\textit{Solvency}, \emptyset)$, 
as introduced in \cref{ex:lending-casino-intro}.
Throughout,
we take the global invariant $I = \mathbb{Z}$ to be the entire state space
and the absorbing region $J_1 = \textit{Debt}$, and verify that this choice
satisfies the premises of \cref{thm:absorbing-regions-decomposition}. Since
$\textit{Debt}$ is disjoint from $\textit{Solvency}$, the return time to
$\textit{Solvency}$ from any $w \in \textit{Debt}$ coincides with the hitting
time computed in \cref{eqn:lending-casino-hitting-probability}, which
establishes \cref{eqn:semantic-1-quant}. For \cref{eqn:semantic-2-quant},
because the states in \textit{Solvency} form a symmetric
one-dimensional random walk, the state $-1$ is reached almost surely from
any non-negative starting point; therefore the chain returns to $B_1 \cup J_1 \cup I^{\sf c} = \emptyset \cup
\textit{Debt} \cup \emptyset$ almost surely from every state in
$I$. Since $\Prob_\mu(I^\omega) = 1$ trivially,
\cref{thm:absorbing-regions-decomposition} yields
$\Prob_\mu(\Fin(\textit{Solvency})) = 1$.

We now establish the same facts via supermartingale certificates. To prove
that $\textit{Debt}$ is an absorbing region with respect to
$\textit{Solvency}$, we invoke \cref{lem:absorbing1} for the certificate $V_1$
shown in \cref{fig:casino-soundness}. By construction, $V_1(s)$ gives the
exact probability of returning to $\textit{Solvency}$ from $s \in \textit{Debt}$, and
its supremum over $\textit{Debt}$ is
$\left( \frac{1 - \varepsilon}{1 + \varepsilon} \right) < 1$
(cf.\ \cref{eqn:lending-casino-hitting-probability}). This discharges the
qualitative safety obligation of \cref{eqn:semantic-1-quant} and confirms that
$J_1 = \textit{Debt}$ is absorbing with respect to $\textit{Solvency}$. It
remains to certify the almost-sure termination obligation
\cref{eqn:semantic-2-quant}, which we resolve in the rest of the section by
applying pre-existing proof rules for almost-sure termination
(cf.\ \cref{ex:lending-casino-rule1,ex:proof-rule-2-lending-casino}). \lipicsEnd
\end{example}

\Cref{cor:inv,lem:absorbing1} discharge the safety obligations of our
decomposition using sound supermartingale certificates that are complete over 
general state spaces. It remains to certify the almost-sure termination obligation
of \cref{eqn:semantic-2-quant} generated by our absorbing-region decomposition, that requires termination into
$B_i \cup J_i \cup I^{\sf c}$ from every state in $I$, for each Streett pair
$(A_i, B_i)$. This obligation may be discharged using any proof rule for
almost-sure termination. Furthermore, if the chosen rule is sound and complete
over a specified class of systems, our framework lifts it to a sound and
$\varepsilon$-complete proof rule for quantitative reactivity over that
class, and to a sound and fully complete one when the reactivity property
holds almost surely. 

We illustrate our methodology using two existing supermartingale
proof rules for almost-sure termination that are known to be complete over
countably infinite state spaces, yielding the first sound and complete
supermartingale proof rules for $\omega$-regular properties on this class
of systems.

\paragraph*{Proof Rule \#1} We discharge the almost-sure termination
obligation (\cref{eqn:semantic-2-quant}) of our absorbing-region decomposition
using the sound and complete proof rule of
\cite[Proof Rule 3.2]{MajumdarS25}. Let $(A_1, B_1), \dots, (A_k, B_k)$ be a
Streett acceptance condition, with $A_i, B_i \subseteq S$ for $i = 1, \dots,
k$, over the countable state space $(S, 2^S)$ of a time-homogeneous Markov
chain with transition kernel $P$. Then, {\em Proof Rule \#1}
requires the existence of
\begin{itemize}
    \item \textbf{a supporting invariant} $I \subseteq S$ with
        $V_0 \colon S \to \Real_{\geq 0}$,
    \item \textbf{absorbing regions} $J_1 \subseteq I \smallsetminus A_1, \dots, J_k \subseteq I \smallsetminus A_k$,
    \item \textbf{supermartingale certificates}
        $V_1, W_1, \dots, V_k, W_k \colon S \to \Real_{\geq 0}$,
    \item \textbf{ranking functions} $U_1, \dots, U_k \colon S \to \Nat$,
\end{itemize}
such that the following conditions hold:
\begin{tcolorbox}[proofrulebox, title={Qualitative safety}]
For every $i = 1, \ldots, k$:
\tightequations\tagpad
\begin{align}
    &\forall s 
    \notin A_i 
    \mathpunct{.} PV_i(s) \leq V_i(s) \label{eqn:V-cond1} \\
    &\forall s\in A_i \mathpunct{.} V_i(s) \geq 1 \label{eqn:V-cond3} \\
    &\exists \gamma_i > 0 \; \forall s \in J_i \mathpunct{.} V_i(s) \leq 1 - \gamma_i \label{eqn:V-cond2}
\end{align}
\end{tcolorbox}

\begin{tcolorbox}[middleproofrulebox, title={Almost-sure termination}]
For every $i = 1, \ldots, k$:
\tightequations\tagpad
\begin{align}
    &\forall s \in I \smallsetminus (B_i \cup J_i)\mathpunct{.} PW_i(s) \leq W_i(s) \label{eqn:majumdar-1} \\
    &\forall s \in I \smallsetminus (B_i \cup J_i) \mathpunct. 
    W_i(s) > 0\\
    &\forall s \in I \cap (B_i \cup J_i)  \mathpunct{.} W_i(s) = U_i(s) = 0 \label{eqn:majumdar-2} \\
    &\forall r \in \mathbb{R}_{\geq 0} \; \exists \varepsilon_r \in \mathbb{R}_{> 0} \; \forall s \in I \smallsetminus (B_i \cup J_i)
    \mathpunct{.} \notag \\[-0.5ex]
    &\qquad W_i(s) \leq r \implies P\big(s, \{ s^\prime : U_i(s^\prime) < U_i(s) \}\big) \geq \varepsilon_r \label{eqn:majumdar-4} \\[1ex]
    &\forall r \in \mathbb{R}_{\geq 0} \mathpunct{.} U_i[\{ s \in I \colon W_i(s) \leq r \}] \subseteq \mathbb{N} \text{ is bounded.} \label{eqn:proof-rule-1-last}
\end{align}
\end{tcolorbox}

\begin{tcolorbox}[proofrulebox, title={Quantitative safety}]
\tightequations\tagpad
\vspace{-2mm}
\begin{align}
    &\forall s \in I \mathpunct{.} PV_{0}(s) \leq V_{0}(s) \label{eqn:proof-rule-1-first} \\
    &\forall s \notin I \mathpunct{.} V_{0}(s) \geq 1 \label{eqn:proof-rule-V0-2nd-cond}
\end{align}
\end{tcolorbox}

We apply the results of 
\cref{sec:decomp} to establish that, over countable Markov chains, {\em Proof Rule \#1} is sound and $\varepsilon$-complete for quantitative reactivity properties, and complete for almost-sure reactivity.

\begin{theorem}[Proof Rule \#1 is Sound]\label{thm:proof-rule-1-sound}
Suppose conditions \eqref{eqn:V-cond1} to \eqref{eqn:proof-rule-V0-2nd-cond} hold.\\
Then $\Prob_\mu\left(\bigcap_{i = 1}^{k} \Fin(A_i) \cup \Inf(B_i)\right) \geq 1 - \mu V_0$.
\end{theorem}
\begin{proof}
    We note that {\em Proof Rule \#1} is an instantiation of the proof rules in \cref{cor:inv,lem:absorbing1} along with the sound proof rule for almost-sure termination from \cite[Proof Rule 3.2 \& Lemma 3.4]{MajumdarS25} which are sufficient conditions for the premises of \cref{thm:absorbing-regions-decomposition}.
\end{proof}

\begin{theorem}[Proof Rule \#1 is $\varepsilon$-Complete]\label{thm:rule1-eps-complete}
Suppose 
that $(\mu, P)$ is a Markov chain on a countable state-space with
$\Prob_\mu\left(\bigcap_{i = 1}^{k} \Fin(A_i) \cup \Inf(B_i)\right) = p$
for some $p \in (0, 1]$.
Then for every $\varepsilon > 0$, conditions \eqref{eqn:V-cond1} to \eqref{eqn:proof-rule-V0-2nd-cond} hold for a suitable supporting invariant and absorbing regions, 
supermartingale certificates, and ranking functions, with $\mu V_0 \leq 1-p+\varepsilon$.
\end{theorem}
\begin{proof}
    By \cref{thm:eps-completeness} there exists a suitable supporting invariant and absorbing regions, such that the almost-sure termination obligation \eqref{eqn:semantic-2-quant} holds, and by the completeness result of \cite[Lemma 3.4]{MajumdarS25} the result follows.
\end{proof}

\begin{theorem}[Proof Rule \#1 is Complete for Almost-Sure Reactivity]\label{eqn:as-react-proof1}
Suppose 
that $(\mu, P)$ is a Markov chain on a countable state-space with
$\Prob_\mu\left(\bigcap_{i = 1}^{k} \Fin(A_i) \cup \Inf(B_i)\right) = 1$.
Then, conditions \eqref{eqn:V-cond1} to \eqref{eqn:proof-rule-V0-2nd-cond} hold for suitable supporting invariant, absorbing regions, supermartingale certificates, and ranking functions, with $\mu V_0 = 0$.
\end{theorem}
\begin{proof}
    We invoke an analogous argument to \cref{thm:rule1-eps-complete}, except we invoke \cref{thm:completeness} instead of \cref{thm:eps-completeness}.
\end{proof}

\begin{example}[Proof Rule \#1 applied to the Lending Casino]\label{ex:lending-casino-rule1}
We continue \cref{eqn:debt-absorbing}, where we established that
$I = \mathbb{Z}$ and $J_1 = \textit{Debt}$ form a valid invariant and
absorbing region, witnessed by the certificates $V_0$ and $V_1$ of
\cref{fig:casino-soundness}. To discharge the almost-sure
termination obligation \eqref{eqn:semantic-2-quant}, Proof
Rule \#1 requires a supermartingale certificate $W_1 \colon \mathbb{Z} \to
\Real_{\geq 0}$ and a ranking function $U_1 \colon \mathbb{Z} \to \Nat$ satisfying
\eqref{eqn:majumdar-1}--\eqref{eqn:proof-rule-1-last} with
$J_1 = \textit{Debt}$ and $B_1 = \emptyset$, for which we exhibit 
    $U_1(s) = \max\{ s+1, 0 \}$ and 
    $W_1(s) = \max\{ s+1, 0 \}$ (cf.\ \cite[Example 3.5]{MajumdarS25}).\lipicsEnd
\end{example}

\paragraph*{Proof Rule \#2} 
We next instantiate our absorbing-region decomposition using the proof rule for almost-sure termination introduced by \cite[Theorem 4.1]{McIverMKK18} (equivalently, \cite[Proof Rule 3.3]{MajumdarS25}), and which is known to be complete over countable-state Markov chains \cite[Theorem 3.11]{MajumdarS25}.
Suppose $(A_1,B_1), \dots, (A_k,B_k)$ 
is a Streett acceptance condition with $A_i,B_i \subseteq S$ for $i = 1 ,\dots, k$
on the countable state space $(S, 2^S)$ of a time-homogeneous Markov chain defined by $(\mu, P)$.
Then, \emph{Proof Rule \#2} requires the existence of 
\begin{itemize}
    \item \textbf{a supporting invariant} $I \subseteq S$ with $V_0 : S \to \Real_{\geq 0}$
    \item \textbf{absorbing regions} $J_1 \subseteq I \smallsetminus A_1, \ldots, J_k \subseteq I \smallsetminus A_k$
    \item \textbf{supermartingale certificates} 
    $V_1, W_1, \ldots, V_k, W_k : S \to \Real_{ \geq 0}$ 
    \item \textbf{decrease functions} $d_1, \ldots, d_k : \Real_{\geq 0} \to \Real_{> 0}$, and 
    \item \textbf{probability functions} $p_1, \ldots, p_k : \Real_{\geq 0} \to (0, 1]$
\end{itemize}
such that the following conditions hold:
\begin{tcolorbox}[proofrulebox, title={Qualitative safety}]
For every $i = 1, \ldots, k$:
\tightequations\tagpad
\begin{align}
    &\forall s\notin A_i \mathpunct{.} PV_i(s) \leq V_i(s) \label{eqn:rule2-c1} \\
    &\forall s\in A_i \mathpunct{.} V_i(s) \geq 1 \label{eqn:rule2-c2} \\
    &\exists \gamma_i > 0 ~ \forall s \in J_i \mathpunct{.} V_i(s) \leq 1 - \gamma_i \label{eqn:rule2-c3}
\end{align}
\end{tcolorbox}

\begin{tcolorbox}[middleproofrulebox, title={Almost-sure termination}]
For every $i = 1, \ldots, k$:
\tightequations\tagpad
\begin{align}
    &\forall s \in I \smallsetminus (B_i \cup J_i) \mathpunct{.} W_i(s) > 0 \label{eqn:rule2-c4} \\
    &\forall s \in I \cap  (B_i \cup J_i)  \mathpunct. W_i(s) = 0\\
    &\forall s \in I \smallsetminus (B_i \cup J_i) \mathpunct{.}
        PW_i(s) \leq W_i(s)\\
    &\forall r \in \Real_{> 0} ~ \forall s \in I \smallsetminus (B_i \cup J_i) \mathpunct{.} \notag \\
    &\qquad (W_i(s) = r)
        \implies
        P(s, \{s^\prime \colon W_i(s^\prime) \leq r - d_i(r) \})
        \geq p_i(r) \label{eqn:rule2-cdrift} \\
    &\forall r_1, r_2 \in \Real_{> 0} \mathpunct{.}
        r_1 \leq r_2 \implies p_i(r_1) \geq p_i(r_2)
        \label{eqn:rule2-c5} \\
    &\forall r_1, r_2 \in \Real_{> 0} \mathpunct{.}
        r_1 \leq r_2 \implies d_i(r_1) \geq d_i(r_2)
        \label{eqn:proof-rule-2-last}
\end{align}
\end{tcolorbox}

\begin{tcolorbox}[proofrulebox, title={Quantitative safety}]
\tightequations\tagpad
\vspace{-2mm}
\begin{align}
    &\forall s \in I \mathpunct{.} PV_{0}(s) \leq V_{0}(s) \label{eqn:proof-rule-2-first} \\
    &\forall s \notin I \mathpunct{.} V_{0}(s) \geq 1 \label{eqn:proof-rule-2-V0-2nd-cond}
\end{align}
\end{tcolorbox}

We employ the results of \cref{sec:decomp} to delineate the properties of Proof Rule \#2.

\begin{theorem}[Proof Rule \#2 is Sound]\label{thm:proof-rule-2-sound}
Suppose conditions \eqref{eqn:rule2-c1} to \eqref{eqn:proof-rule-2-V0-2nd-cond} hold.
Then $\Prob_\mu\left(\bigcap_{i = 1}^{k} \Fin(A_i) \cup \Inf(B_i)\right) \geq 1 - \mu V_0$.
\end{theorem}
\begin{proof}
    In the same manner as \cref{thm:proof-rule-1-sound}, we have invoked an existing proof rule \cite[Theorem 4.1]{McIverMKK18} for condition \eqref{eqn:semantic-2-quant}, and the proof follows by \cref{thm:absorbing-regions-decomposition}.
\end{proof}

\begin{theorem}[Proof Rule \#2 is $\varepsilon$-Complete]\label{thm:rule2-eps-complete}
Suppose 
that $(\mu, P)$ is a Markov chain on a countable state-space for which
$\Prob_\mu\left(\bigcap_{i = 1}^{k} \Fin(A_i) \cup \Inf(B_i)\right) = p
$
for some $p \in (0, 1]$.
Then for every $\varepsilon > 0$, conditions 
\eqref{eqn:rule2-c1}
to 
\eqref{eqn:proof-rule-2-V0-2nd-cond}
hold for suitable supporting invariant, absorbing regions, 
supermartingale certificates, decrease and probability functions, with $\mu V_0 \leq 1-p+\varepsilon$.
\end{theorem}
\begin{proof}
    We apply the same argument as \cref{thm:rule1-eps-complete}, except relying on \cite[Theorem 4.1]{McIverMKK18} for the soundness of conditions \eqref{eqn:rule2-c4} to \eqref{eqn:proof-rule-2-last}.
\end{proof}

\begin{theorem}[Proof Rule \#2 is Complete for Almost-Sure Reactivity]\label{thm:rule-2-complete}
Suppose that 
$
\Prob_\mu\left(\bigcap_{i = 1}^{k} \Fin(A_i) \cup \Inf(B_i)\right) = 1.
$
Then,  
conditions 
\eqref{eqn:rule2-c1}
to 
\eqref{eqn:proof-rule-2-V0-2nd-cond} hold for suitable supporting invariant, absorbing regions, supermartingale certificates, decrease and probability functions, with $\mu V_0 = 0$.
\end{theorem}
\begin{proof}
    Similarly, we apply the same argument as \cref{eqn:as-react-proof1}, except relying on \cite[Theorem 3.11]{MajumdarS25} for the completeness of the conditions \eqref{eqn:rule2-c4} to \eqref{eqn:proof-rule-2-last}.
\end{proof}

\begin{example}[Proof Rule \#2 applied to the Lending Casino]\label{ex:proof-rule-2-lending-casino}
Continuing \cref{ex:lending-casino-rule1} and using $I = \mathbb{Z}$, $J_1 = \textit{Debt}$, and the safety certificates of \cref{fig:casino-soundness}, we additionally exhibit $W_1(s) = \max\{s + 1, 0\}$ for conditions \eqref{eqn:rule2-c4} to \eqref{eqn:proof-rule-2-last} within Proof Rule \#2, paired with the constant decrease and probability functions $d_1(r) = 1$ and $p_1(r) = \frac{1}{2}$ (cf.\ \cite[p.8]{McIverMKK18}).\lipicsEnd
\end{example}

\begin{remark}[Prior Supermartingale Certificates for $\omega$-Regular Properties are Incomplete]
We note that all prior supermartingale certificates for $\omega$-regular properties in the literature, which meet the requirement that they are defined in terms of value functions whose codomain is real-valued \cite{AbateGR25,DBLP:conf/cav/AbateGR24,HenzingerMSZ25,DBLP:conf/tacas/ChakarovVS16}, are inapplicable to proving the almost-sure satisfaction of the reactivity property $\Fin(\textit{Solvency})$ explored in \cref{eqn:debt-absorbing}, due to the fact that the time to reach \textit{Debt} from any state in \textit{Solvency} has infinite expected value, which is beyond the reach of prior proof rules (cf.\ \cite[Example 4.12]{DBLP:journals/corr/abs-2512-00270}). \lipicsEnd
\end{remark}


\section{Related Work}\label{sec:relwork}

Supermartingale proof rules have been developed as tools to analyse transience and recurrence on aperiodic and irreducible Markov chains, and applied to draw conclusions about their long-term behaviour, such as their stationary distributions (cf.\ \cite{Foster1953,Pakes1969} and \cite[Theorem 11.0.1]{meyn_tweedie_glynn_2009}).
 For countable-state Markov chains that are aperiodic and irreducible, \cite{Mertens78} introduced sound and complete supermartingale proof rules for almost-sure recurrence \cite[Theorem 1]{Mertens78} (cf.~\cite[Theorem 7.5.2]{douc2018markov}) and almost-sure transience \cite[Theorem 2]{Mertens78} (cf.~\cite[Theorem 7.5.1]{douc2018markov}). Note that for irreducible aperiodic chains on $\Nat$, recurrence corresponds to visiting $\{0\}$ infinitely often, and transience to visiting it finitely often.

 Although transience and recurrence are closely related to the events $\Fin(A)$ and $\Inf(B)$ that constitute reactivity properties (\cref{def:reactivity}), we identify two major reasons why these classic results on necessary and sufficient conditions for almost-sure transience and almost-sure recurrence do not easily transfer to reactivity verification on countably infinite Markov chains.
Firstly, they require irreducibility, meaning that a finite path exists between any pair of states. The Markov chain induced by synchronous composition (cf.\ \cref{rem:omega-regular}) with a deterministic Streett automaton $\mathcal{A}$ fails to be irreducible, because $\mathcal{A}$ in general contains more than one strongly connected component, regardless of the irreducibility of the underlying Markov chain $\hat{\chain}$. Despite the fact that countable Markov chains admit a partitioning into transient and recurrent components (\cite[Section 4.5]{meyn_tweedie_glynn_2009}, \cite[Theorem 7.3.3]{douc2018markov} and \cite{Doblin1940}), the transient (and thus potentially non-irreducible) component may encompass the entire chain.
Secondly, every reactivity property is a Boolean combination of $\Fin$- and $\Inf$-events, and can therefore be satisfied almost surely without any of its constituent $\Fin$- or $\Inf$-events being satisfied almost surely. Proof rules for almost-sure transience and almost-sure recurrence, applied in isolation, are therefore too conservative for reactivity properties represented by deterministic Streett and parity automata \cite[Figure 1]{DBLP:conf/cav/AbateGR24}.

Subsequent works have relaxed the irreducibility assumption of \cite{Mertens78} for special cases of almost-sure reactivity.
In the case of the reactivity property defined by a single Streett pair $(A, \emptyset)$, \cite[Proposition 8.3.1(iv)]{meyn_tweedie_glynn_2009} defines a proof rule for uniform transience on Markov chains without assuming irreducibility. While this condition is similar in spirit to our requirement in \cref{eqn:semantic-1-quant}, almost-sure uniform transience is strictly stronger than requiring $\Fin(A)$ to hold with probability 1.
For instance, the region \textit{Solvency} in the Lending Casino (\cref{fig:lending-casino}) is not uniformly transient, even though the event $\Fin(\textit{Solvency})$ holds almost surely (see \cref{eqn:debt-absorbing}). 
Dually, by considering the reactivity property defined by the Streett pair $(S, \textit{Halt})$---where every state in $\textit{Halt} \subseteq S$ is made absorbing---one recovers the property of probabilistic termination. Under this instantiation, our framework subsumes the known result that quantitative termination is equivalent (up to $\varepsilon$-approximation) to a combination of almost-sure termination and quantitative safety obligations (cf.\ \cite[Lemma 4.6]{MajumdarS25},
\cite{DBLP:conf/cav/ChatterjeeGMZ22},
and \cite[Example 4]{AbateGR25}).

While we consider reactivity properties defined upon Markov chains, prior work has treated questions of completeness for recurrence and transience properties defined upon countable MDPs under adversarial and controllable interpretations of non-determinism, as well as for stochastic games \cite{DBLP:conf/icalp/KieferMST19,DBLP:conf/sigecom/KieferMST25}. In the case of controllable non-determinism, prior results demonstrate the necessary existence of policies for the attainment of reactivity conditions \cite{DBLP:conf/concur/KieferMST20,DBLP:conf/concur/KieferMST21} and show that these may be constructed from optimal policies for simpler obligations (such as reachability, safety, and recurrence). It is an open question as to whether our decomposition results (\cref{sec:decomp}) extend to MDPs in any form.

The algorithmic synthesis of supermartingale certificates draws on several research areas: ranking functions and inductive invariants in program analysis~\cite{DBLP:books/daglib/0080029,DBLP:conf/cav/ColonSS03,DBLP:journals/scp/ErnstPGMPTX07,DBLP:journals/tosem/NguyenKWF14,DBLP:conf/sas/SankaranarayananSM04,ProbTermToolAmber2021,DBLP:journals/pacmpl/MoosbruggerSBK22,DBLP:journals/toplas/TakisakaOUH21}, Lyapunov functions and barrier certificates in control theory~\cite{DBLP:conf/cdc/PrajnaJP04,DBLP:conf/cdc/PrajnaPP02,DBLP:conf/cdc/Papachristodoulou02,sankaranarayanan2013LyapunovFunctionSynthesis,she2013DiscoveringPolynomialLyapunov}, and martingale-like value functions underpinning reinforcement learning~\cite{DBLP:books/lib/SuttonB98}.
Across these domains, supermartingale certificates are typically expressed as real-valued functions and synthesized either by deductive symbolic reasoning or by inductive learning. Symbolic methods compute linear or polynomial certificates via Farkas' Lemma~\cite{AgrawalC018,DBLP:conf/cav/ChakarovS13,DBLP:journals/toplas/ChatterjeeFNH18,DBLP:conf/popl/ChatterjeeNZ17,DBLP:conf/tacas/ColonS01,DBLP:conf/cav/ColonS02} and Positivstellensatz results~\cite{DBLP:conf/cav/ChatterjeeFG16,ChatterjeeAAAI2025,DBLP:conf/pldi/AsadiC0GM21,DBLP:conf/fm/ChatterjeeGGKZ24,HandelmanOriginalPaper1988,DBLP:journals/pacmpl/ChatterjeeGNZ24,DBLP:conf/pldi/ZikelicCBR22,PLDI24HandelmanBayesian}, reducing synthesis to decision problems in the first-order theory of the reals.
Learning-based techniques instead infer certificates from parametric families via counterexample-guided inductive synthesis~\cite{DBLP:conf/tacas/BatzCJKKM23,DBLP:conf/tacas/AhmedPA20}, often using neural networks as their representation~\cite{DBLP:conf/cav/AbateGR20,DBLP:conf/aaai/ZikelicLHC23,DBLP:conf/atva/AnsaripourCHLZ23,DBLP:conf/nips/ZikelicLVCH23,DBLP:conf/tacas/ChatterjeeHLZ23,DBLP:conf/aaai/LechnerZCH22,ctsm,DBLP:conf/nips/GiacobbeKPT24,DBLP:conf/nips/GiacobbeKPT25,DBLP:journals/lmcs/AbateEGPR26}.

\section{Conclusion}
We have established that every reactivity property defined upon a Markov chain over a general state space admits a decomposition into multiple absorbing regions and a single global invariant. This decomposition, in turn, yields a strategy for proving probabilistic reactivity that reduces the problem to solving obligations of qualitative safety, quantitative safety, and almost-sure termination. Our technical results in \cref{thm:absorbing-regions-decomposition,thm:eps-completeness} establish that this strategy is sound and 
$\varepsilon$-complete in the general case, while \cref{thm:completeness} shows that it is fully complete when the reactivity property holds almost surely. 

Our strategy yields the first proof rules involving scalar-valued supermartingale certificates that are both sound and complete for reactivity properties (\cref{sec:complete-certificates}); this provides a stepping stone towards automation through the adaptation of existing methods for scalar-valued certificates (\cref{sec:relwork}), which remains an open question and a subject for future work. These new supermartingale certificates for reactivity are obtained by instantiating our strategy with existing proof rules for almost-sure termination that are sound and complete in the special case of countably infinite state spaces.

Our decomposition results (\cref{thm:eps-completeness,thm:completeness,fig:rule_stack}) apply to general state spaces, whereas the concrete supermartingale certificates for reactivity we present currently possess completeness guarantees in the setting of countably infinite state spaces.
Extending these results to uncountable state spaces is an interesting direction for future work. Our results show that addressing this challenge reduces to that of certifying almost-sure termination over general state spaces in a complete fashion, from which a complete proof rule for reactivity over general state spaces would follow (cf.\ \cref{fig:rule_stack}). Moreover, our strategy assumes time homogeneity; extending its scope of applicability to Markov decision processes over general state spaces is another promising direction for future research.

\bibliography{main}

\end{document}